\documentclass[11pt]{article}

\advance \oddsidemargin by -0.8in
\newcommand{\myparskip}{3pt}
\parskip \myparskip

\usepackage{natbib}
 \bibpunct[, ]{(}{)}{,}{a}{}{,}%
 \def\BIBand{and}%
\usepackage{amsfonts}
\usepackage{amsmath}
\usepackage{amsthm}
\usepackage{graphicx}
\usepackage{algorithm}
\usepackage{algorithmic}
\usepackage{color}
\usepackage{subfigure}
\usepackage{makecell}
\allowdisplaybreaks[4]

\newtheorem{thm}{Theorem}[section]
\newtheorem{theorem}[thm]{Theorem}
\newtheorem{lemma}[thm]{Lemma}
\newtheorem{proposition}[thm]{Proposition}
\newtheorem{corollary}[thm]{Corollary}
\newtheorem{definition}[thm]{Definition}

\def\reff#1{(\ref{#1})}
\def\nn{\nonumber}

\def\veps{\varepsilon}
\def\bR{{\mathbb{R}}}  
\def\cC{{\mathcal{C}}}
\def\cF{{\mathcal{F}}}
\def\cX{{\mathcal{X}}}
\def\cent{{\rm cent}}
\def\cost{{\rm cost}}
\def\outlier{{\rm outlier}}
\def\bea{\begin{eqnarray}}
\def\eea{\end{eqnarray}}
\def\ba{\begin{array}}
\def\ea{\end{array}}

\DeclareMathOperator*{\argmax}{argmax}  
\DeclareMathOperator*{\argmin}{argmin}

\makeatletter
\newenvironment{breakablealgorithm}
  {
   \begin{center}
     \refstepcounter{algorithm}
     \hrule height.8pt depth0pt \kern2pt
     \renewcommand{\caption}[2][\relax]{
       {\raggedright\textbf{\ALG@name~\thealgorithm} ##2\par}%
       \ifx\relax##1\relax 
         \addcontentsline{loa}{algorithm}{\protect\numberline{\thealgorithm}##2}%
       \else 
         \addcontentsline{loa}{algorithm}{\protect\numberline{\thealgorithm}##1}%
       \fi
       \kern2pt\hrule\kern2pt
     }
  }{
     \kern2pt\hrule\relax
   \end{center}
  }
\makeatother

\begin{document}

\title{Outliers Detection Is Not So Hard: Approximation Algorithms for
Robust Clustering Problems Using Local Search Techniques}

\author{
Yishui Wang\thanks{School of Mathematics and Physics, University of Science and Technology Beijing, Beijing 100083, P.R. China. Email: {\tt wangys@ustb.edu.cn}.} \and
Rolf H. M${\ddot{\rm o}}$hring\thanks{Institute for Applied Optimization, Department of Computer Science and Technology, Hefei University, P.R. China, and The Combinatorial Optimization and Graph Algorithms (COGA) group, Institute for Mathematics, Technical University of Berlin, Germany. Email: {\tt Rolf.Moehring@tu-berlin.de}.}
\and
Chenchen Wu\thanks{Corresponding author. College of Science, Tianjin University of Technology,  Tianjin 300384, P.R. China. Email: {\tt wu\_chenchen\_tjut@163.com}.}
\and
Dachuan Xu\thanks{Department of Operations Research and Information Engineering, Beijing University of Technology, Beijing 100124, P.R. China. Email: {\tt xudc@bjut.edu.cn}.}
\and
Dongmei Zhang\thanks{School of Computer Science and Technology, Shandong Jianzhu University, Jinan 250101, P.R. China. Email: {\tt zhangdongmei@sdjzu.edu.cn}.}
}


\begin{titlepage}
\maketitle
\thispagestyle{empty}

\begin{abstract}
In this paper, we consider two types of robust models of the $k$-median/$k$-means problems: the outlier-version ($k$-MedO/$k$-MeaO) and the penalty-version ($k$-MedP/$k$-MeaP), in which we can mark some points as outliers and  discard them. In $k$-MedO/$k$-MeaO, the number of outliers is bounded by a given integer. In $k$-MedP/$k$-MeaP, we do not bound the number of outliers, but each outlier will incur a penalty cost.
We develop a new technique to analyze the approximation ratio of local search algorithms for these two problems by introducing an adapted cluster that can capture useful information about outliers in the local and the global optimal solution. For $k$-MeaP, we improve the best known approximation ratio based on local search from $25+\veps$ to $9+\veps$. For $k$-MedP, we obtain the best known approximation ratio. For $k$-MedO/$k$-MeaO, there exists only two bi-criteria approximation algorithms based on local search. One violates the outlier constraint (the constraint on the number of outliers), while the other violates the cardinality constraint (the constraint on the number of clusters). We consider the former algorithm and improve its approximation ratios from $17+\veps$ to $3+\veps$ for $k$-MedO, and from $274+\veps$ to $9+\veps$ for $k$-MeaO.
\end{abstract}
\end{titlepage}

\section{Introduction} \label{sec:introduction}
Using large data sets to make better decisions is becoming more important
and routinely applied in
Operations Research, Management Science, Biology, Computer Science,
and Machine Learning
(see e.g. \citealp{bms,bh,hl,lw}).
Clustering large data is a fundamental problem in data analytics.
Among many clustering types, center-based clustering
is the most popular and widely used one.
Center-based clustering problems include
$k$-median, $k$-means, $k$-center, facility location problems, and so on (see e.g. \citealt{answ,bprst,lloyd,li,lswxxz,nssxz}).
The $k$-median and $k$-means problems are the most basic and classic clustering problems.
The goal of $k$-median/means clustering
is to find $k$ cluster centers such that
the total (squared) distance from each input datum to the closest cluster center
is minimized.
Usually, one considers $k$-median problems in arbitrary metrics
while $k$-means problems in the Euclidean space $\mathbb{R}^D$.

Both problems are NP-hard to approximate beyond  with the lower bounds
$1 + 2/e \approx 1.736$ (\citealp{jmmsv})
and $1.07$ (\citealp{ck})
for $k$-median and $k$-means, respectively.
There are many papers on designing efficient approximation algorithms.
The best known approximations are $2.675+\varepsilon$ (\citealp{bprst}) and
$6.357+\varepsilon$ (\citealp{answ})
for $k$-median and $k$-means, respectively.
If we restrict to a fixed-dimensional Euclidean space,
the $k$-median and $k$-means problems have a PTAS (see \citealp{arr,ckm,frs}).

However, real-world data sets often contain outliers
which may totally spoil $k$-median/means clustering results.
To overcome this problem, robust clustering techniques have been developed
to avoid  being affected by outliers. In general, there are two types
of robust formulations:
$k$-median/means with outliers ($k$-MedO/$k$-MeaO) and
$k$-median/means with penalties ($k$-MedP/$k$-MeaP).
We formally define these problems as follows.

\begin{definition}[$k$-Median Problem with Outliers/Penalties] \label{defmedian}
In the $k$-median problem with outliers ($k$-MedO),
we are given a client set $ \mathcal{X} $  of $n$ points,
a facility set $ \mathcal{F} $ of $m$ points,
a metric space $(\mathcal{X} \cup \mathcal{F}, d)$,
and two positive integers $k < m $ and $z < n $.
The aim is to find a subset
$ S \subseteq  \mathcal{F} $ of cardinality at most $ k $,
and an outlier set $ P \subseteq \mathcal{X} $ of cardinality at most $ z $
such that the objective function
$ \sum_{x \in \mathcal{X} \setminus P } \min_{s \in S } d(x, s)   $
is minimized.
In the $k$-median problem with penalties ($k$-MedP),
we have the same input except that
the cardinality restrictions on the penalty set
$P \subseteq \mathcal{X}$ is replaced by
a nonnegative penalty $p_x $ for each $ x \in \mathcal{X} $,
and the objective function is to minimize
$ \sum_{x \in \mathcal{X} \setminus P } \min_{s \in S } d(x, s)
+  \sum_{x \in  P } p_x $.
\end{definition}

\begin{definition}[$k$-Means Problem with Outliers/Penalties] \label{defmeans}
In the $k$-means problem  with outliers ($k$-MeaO),
we are given a data set $ \mathcal{X} $ in $ \mathbb{R}^d$ of $n$ points
and two positive integers $k$ and $z < n $. Let $d(u,v):=\|u-v\|_2$ be the Euclidean distance of the points $u$ and $v$.
The aim is to find a cluster center set
$ S \subseteq  \mathbb{R}^d$ of cardinality at most $ k $,
and an outlier set $ P \subseteq \mathcal{X} $ of cardinality at most $ z $
such that the objective function
$ \sum_{x \in \mathcal{X} \setminus P } \min_{s \in S } d(x,s)^2  $
is minimized.
In the $k$-means problem with penalties ($k$-MeaP),
we have the same input except that
the cardinality restrictions on penalty set
$P \subseteq \mathcal{X}$  is instead of
a nonnegative penalty $p_x $ for each $ x \in \mathcal{X} $,
and the objective function is to minimize
$ \sum_{x \in \mathcal{X} \setminus P } \min_{s \in S } d(x,s)^2
+  \sum_{x \in  P } p_x $.
\end{definition}

From the perspective of clustering, we can view a facility in the $k$-median problem as a center, and view a client as a point. To avoid confusion, we use ``center'' and ``point'' for all problems we consider in this paper.

Basic versions of these problems have been widely studied, and many approximation algorithms based on many different techniques, including LP-rounding (see e.g. \citealp{cgts,chl,li}), primal-dual (see e.g. \citealp{answ,jv}), dual-fitting (see e.g. \citealp{jmmsv,myz}), local search (see e.g. \citealp{agkmmp,kmnpsw,kpr}), Lagrangian relaxation (see e.g. \citealp{jmmsv,jv}), bi-point rounding (see e.g. \citealp{jmmsv,jv}), and pseudo-approximation (see e.g. \citealp{bprst,ls}), have been developed and applied.

We now discuss the state-of-art approximation results for robust version of $k$-median/$k$-means. \cite{chen} has presented the first constant but very large  approximation algorithm for $k$-MedO via successive local search.
\cite{kls} have obtained an iterative LP rounding framework yielding
 $(7.081 +\varepsilon)$- and $(53.002 +\varepsilon)$-approximation algorithms
for $k$-MedO and $k$-MeaO, respectively.
To the best of our knowledge, these are only two constant factor approximation results for $k$-MedO and $k$-MeaO.

The first constant $4$-approximation for $k$-MedP has been given by \cite{ckmn}
using Lagrangian relaxation framework of \cite{jv}.
The best  $(3+\varepsilon)$-approximation for $k$-MedP has been obtained by \cite{hkk}, who called the problem the red-blue median problem.
Three years later, \cite{wxdw} have independently obtained the same factor approximation for $k$-MedP
and have further generalized it to a $(3.732+\varepsilon)$-approximation for the $k$-facility location problem with linear penalties, which is a common generalization of facility location (in which there are facility opening costs and no cardinality constraint) and $k$-MedP.
Both of them use local search techniques. \cite{z} has obtained  the approximation ratio $3.732+\varepsilon$ for the $k$-facility location problem ($k$-FLP).
The currently best ratio of $3.25$ for $k$-FLP is due to \cite{chl}.
For the $k$-median problem with uniform penalties,
\cite{wdx} have adapted the pseudo-approximation technique of \cite{ls} and obtained a $(2.732+\varepsilon)$-approximation.

\cite{zhwxz} have presented the first constant $(25+\varepsilon)$- approximation algorithm for $k$-MeaP using local search.
\cite{fzsw} have improved this to a $(19.849+\varepsilon)$-approximation by combing Lagrangian relaxation with bipoint rounding.
A summary of the up-to-date approximation results for $k$-MedO/$k$-MeaO and $k$-MedP/$k$-MeaP along with their ordinary versions is given in Table \ref{tab1}.

\begin{table}[h] \label{tab1}
\fontsize{8.5}{11.4}\selectfont
\centering
\caption{Comparion of (robust) clustering problems.}
\vspace{0.2cm}
\begin{tabular}{|l|l|l|l|l|l|l|}
\hline
Techniques and reference & $k$-median & $k$-MedO & $k$-MedP & $k$-means & $k$-MeaO & $k$-MeaP 
\\  \hline
LP rounding \citep{cgts}   &  $6\frac23$  &    &    &  &   &
\\  \hline
Lagrangian relaxation \citep{jv}   &  $6$  &    &    & $108$  &   &
\\  \hline
Lagrangian relaxation \citep{ckmn}  &    &    & $4$   &  &   &
\\  \hline
Lagrangian relaxation \citep{jmmsv}  &  $4$  &    &    &   &   &
\\  \hline
Local search \citep{agkmmp}  & $3+\varepsilon$  &   &   &  &   &
\\  \hline
Local search \citep{kmnpsw}  &    &   &   & $9+\varepsilon$   &   &
\\ \hline
Successive local search \citep{chen}  &    & constant  &    &     &    &
\\ \hline
Dependent LP rounding \citep{chl}   & $3.25 $    &    &    &     &    &
\\ \hline
Local search \citep{hkk} &    &   & $3+\varepsilon$   &     &    &
\\ \hline
Pseudo-approximation \citep{ls} & $2.732+\varepsilon$    &   &     &     &    &
\\ \hline
Pseudo-approximation \citep{bprst} & $2.675+\varepsilon$  &   &   &  &   &
\\ \hline
Iterative LP rounding \citep{kls} &    & $7.081 +\varepsilon $  &   &     & $53.002 +\varepsilon $  &
\\ \hline
Primal-dual  \citep{answ}  &    &   &   & $6.357+\varepsilon$   &   &
\\  \hline
Local search \citep{zhwxz}   &    &   &   &    &   & $25+\varepsilon$
\\  \hline
Bipoint rounding \citep{fzsw}   &    &   &   &  &   & $19.849+\varepsilon$
\\  \hline
\end{tabular}
\end{table}

The available literature suggests two observations concerning the approximation factor:
i) $k$-MedP/$k$-MeaP seems more easy to approximate than $k$-MedO/$k$-MeaO.
ii) The existence of outliers make the approximation of the corresponding robust clustering problems much harder than the ordinary clustering problems.

The best known approximation ratios for $k$-MedO and $k$-MeaO have been obtained by LP-rounding, but these algorithms are not strongly polynomial-time since they involve solving linear programs. Concerning time complexity, local search is better than LP-rounding, and this technique has been well applied to $k$-median/$k$-means and their penalty versions. Furthermore, the standard local search algorithm is also used for $k$-median/$k$-means with some special metrics such as the minor-free metric \citep{ckm} and the doubling metric \citep{frs}.  These two papers show that the standard local search scheme yields a PTAS for the considered problems when the dimension is fixed. Their results hold in particular for the Euclidean metric, since both the minor-free metric and the doubling metric are extensions of the Euclidean metric.

Unfortunately, the standard local search algorithm for $k$-MedO/$k$-MeaO can not produce a feasible solution with a bounded approximation ratio \citep{fkrs}. So some research directions focus on bi-criteria approximation algorithms based on local search for these two problems. These algorithms have a bounded approximation ratio but violate either the $k$-constraint or the outlier constraint by a bounded factor. \cite{gklmv} have developed a method for addressing outliers in a local search algorithm, yielding  a bi-criteria $(274+\veps,O(\frac{k}{\veps}\log n\delta))$-approximation algorithm ($\delta$ as defined in Section \ref{sec:contributions}) that violates the outlier constraint. \cite{fkrs} have provided $(3+\veps,1+\veps)$- and $(25+\veps,1+\veps)$-local search bi-criteria  approximation algorithms for $k$-MedO and $k$-MeaO respectively.

We will consider the standard local search algorithm for $k$-MedP/$k$-MeaP, and the outlier-based local search algorithm by \cite{gklmv} for $k$-MedO/$k$-MeaO. Using our new technique, we will improve the approximation ratios for $k$-MeaP, $k$-MeaO and $k$-MedO. For $k$-MedP, we  obtain the same approximation ratio which is the best one possible.

We list the related results about local search algorithms for $k$-MedO/$k$-MeaO and $k$-MedP/$k$-MeaP in Table \ref{tab2}.

\begin{table}[t] \label{tab2}
\fontsize{8.5}{11.4}\selectfont
\centering
\caption{Local search algorithms for (robust) clustering problems. The \# centers blowup means the factor by which the cardinality constraint is violated. The \# outliers blowup means the factor by which the outlier constraint is violated. }
\vspace{0.2cm}
\begin{tabular}{|c|c|c|c|c|}
\hline
Reference  & Problem  &  Ratio  & \# centers blowup & \# outliers blowup 
\\  \hline
\cite{agkmmp}  &  $k$-median  & $3+\varepsilon$  & none & none
\\  \hline
\cite{kmnpsw} &  $k$-means   &  $9+\varepsilon$ & none & none
\\ \hline
\cite{chen}  &  $k$-MedO  & constant & none & none
\\ \hline
\cite{hkk}  & $k$-MedP & $3+\varepsilon$  & none & none
\\ \hline
\cite{zhwxz}  & $k$-MeaP    & $25+\varepsilon$ & none & none
\\  \hline
\cite{ckm} & \makecell[c]{$k$-median/$k$-means \\ in minor-free metrics \\ with fixed dimension}  & PTAS & none & none
\\  \hline
\cite{frs} & \makecell[c]{$k$-median/$k$-means with \\ fixed doubling dimension} & PTAS & none & none
\\  \hline
\cite{fkrs}  & \makecell[c]{$k$-MedO \\ $k$-MeaO} & \makecell[c]{$3+\veps$ \\ $25+\veps$} & \makecell[c]{$1+\veps$ \\ $1+\veps$} & \makecell[c]{none \\ none}
\\  \hline
\cite{gklmv} & \makecell[c]{$k$-MedO \\ $k$-MeaO} & \makecell[c]{$17+\veps$ \\ $274+\veps$} & \makecell[c]{none \\ none} & \makecell[c]{$O(k\log(n\delta)/\veps)$ \\ $O(k\log(n\delta)/\veps)$}
\\  \hline
Our results & \makecell[c]{$k$-MedP \\ $k$-MeaP \\ $k$-MedO \\ $k$-MedO \\ $k$-MeaO \\ $k$-MeaO} & \makecell[c]{$3+\varepsilon$ \\ $9+\varepsilon$ \\ $5+\varepsilon$ \\ $3+\varepsilon$ \\ $25+\varepsilon$ \\ $9+\varepsilon$} & \makecell[c]{none \\ none \\ none \\ none \\ none \\ none} & \makecell[c]{none \\ none \\ $O(k\log(n\delta)/\veps)$ \\ $O(k^2\log(n\delta)/\veps)$ \\ $O(k\log(n\delta)/\veps)$ \\ $O(k^2\log(n\delta)/\veps)$}
\\  \hline
\end{tabular}
\end{table}

\subsection{Our techniques}\label{sec:our-tech}
We concentrate on $k$-MedP and $k$-MeaP to illustrate our techniques.
The associated outlier versions are then easy generalizations.

In the standard local search algorithm, one starts from an arbitrary feasible solution.
Operations such as add center, delete center, or swap centers, define
the neighborhood of the currently  feasible solution.
One then searches for a local optimal solution in the whole neighborhood and takes it as the new  current solution. This is iterated until the improvement becomes sufficiently small.

Similar to the previous analyses of local search algorithms for $k$-median and $k$-means, we want to find some valid inequalities by constructing swap operations in order to establish some ``connections'' between  local and global optimal solutions.
Integrating all these inequalities or connections carefully, we can bound cost of the local optimal solution  by the global optimal cost.

In the analysis of $k$-median (see \citealt{agkmmp}), these connections are given individually for each point (that is, each point yields an inequality that gives a bound of its cost after the constructed swap operation). We call this type of analysis an ``individual form''.

Another analysis type is  the ``cluster form'', in which the connections between the local and global optimal solutions are revealed for some clusters containing several points. The cluster form analysis was first used for $k$-means in \cite{kmnpsw}. In the work of \cite{kmnpsw}, the authors use the Centroid Lemma (introduced in Section \ref{sec:tech-lem}) to obtain equality for each cluster in the optimal solution, and then deduce the approximation ratio by these equalities and the triangle inequality. They found that the cluster form analysis is tighter than the individual form.
However, the same analysis does not apply to $k$-MeaP due to the existence of outliers.
Indeed, the clusters derived with equalities from the Centroid Lemma should contain no outliers in both the local and global optimal solutions, since they do not incur a cost for outliers.

To this end, we careful recognize and define an \emph{adapted cluster} as a cluster that excludes outliers. In order to use the Centroid Lemma, we identify a new centroid for the adapted cluster and use the triangle inequality for the squared distances to identify the associated centroid of the global solution in the analysis. These new centroids can be found by a carefully defined mapping function.

Our cluster form analysis also applies to $k$-MedP, although there is no result like the Centroid Lemma for this problem. In fact, we only need to denote the optimal center of the adapted cluster for $k$-MedP (corresponding to the centroid in the analysis for $k$-MeaP), and use its optimality to derive the inequality for the adapted cluster. During the entire process of the analysis, we do not compute the optimal center, so we do not need a result like the Centroid Lemma.

Our cluster form analysis establishes a bridge between local and global solutions for both robust and ordinary clusterings, and we obtain a clear and unified understanding of them.
Furthermore, we believe that our technique can be generalized to other robust clustering problems such as the robust facility location and $k$-center problems.

\subsection{Our contributions}\label{sec:contributions}
We use the standard local search algorithm for  $k$-MedP and $k$-MeaP. Via a subtle cluster form analysis, we obtain the following result.

\begin{theorem}
The standard local search algorithm yields $(3+\varepsilon)$- and $(9+\varepsilon)$-approximations for $k$-MedP and $k$-MeaP respectively.
\end{theorem}

Our analysis is different to that of \cite{hkk} who have also obtained a local search $(3+\varepsilon)$-approximation for $k$-MedP, and improve the previous local search $(25+\veps)$-approximation \citep{zhwxz} and the primal-dual $(19.849+
\veps)$-approximation \citep{fzsw} for $k$-MeaP. Moreover, our result indicates that the penalty-version of the clustering problems have the same approximation ratios as the ordinary  version, when we adopt the local search technique followed with our cluster form analysis.

For $k$-MedO and $k$-MeaO, we use the outlier-based local search algorithm (based on \citealt{gklmv}).

The  algorithm has a parameter for controlling the descending step-length of the cost in each iteration. This parameter is fixed in \cite{gklmv}, while it is an input in our algorithm  because both the approximation ratio and the number of outliers blowup are associated with the value of this parameter. This helps us to reveal a tradeoff between the approximation ratio and the outlier blowup. When selecting appropriate values for this parameter, we can obtain constant approximation ratios. In the following theorems, $\delta$ denotes the maximal distance between two points in the data set.

\begin{theorem}
The outlier-based local search algorithm yields bicriteria $(5+\varepsilon,O($ $\frac{k}{\veps}\log(n\delta)))$- and $(3+\varepsilon,O(\frac{k^2}{\veps}\log(n\delta)))$-approximations for $k$-MedO, and bicriteria $(25+\veps,O(\frac{k}{\veps}\log(n\delta)))$- and $(9+\veps,O(\frac{k^2}{\veps}\log(n\delta)))$-approximations for $k$-MeaO, where $O(\frac{k}{\veps}\log(n\delta))$ and $O(\frac{k^2}{\veps}\log(n\delta))$ are the factors by which the outlier constraint is violated.
\end{theorem}

With the same outlier blowup, our ratios obtained with single-swap  significantly improve the previous ratios $17 +\varepsilon $  and  $274 +\varepsilon $  for the $k$-MedO and $k$-MeaO, respectively. The multi-swap version improves this even more, but with a larger outlier blowup.

These results strengthens our comprehension of robust clustering problems from a local search aspect. Furthermore, our cluster form analysis  has a high potential to be applied in the robust version for FLP and $k$-FLP, since the structures of these two problems are similar to $k$-MedP, and the analyses for  the connection cost and facility opening cost are seperated in the previous papers that study  local search algorithms for FLP and $k$-FLP (see \citealt{agkmmp, z}).

\subsection{Outline of the paper}
Section 2 presents the unified models and notations for $k$-MedP/$k$-MeaP and $k$-MedO/$k$-MeaO, and some useful technical lemmas. Section 3 then presents our standard local search algorithms for  $k$-MedP/$k$-MeaP and our corresponding theoretical results. In Section 4, we develop our outlier-based local search algorithms for  $k$-MedO/$k$-MeaO and present our corresponding theoretical results. The conclusions are given in Section 5.
All technical proofs are given in the appendices.

\section{Preliminaries}\label{sec:preliminaries}
\subsection{The models}\label{sec:model}
We use the following notation for the problems studied in this paper (in addition to the notation introduced in the introduction). $\cC$ denotes the candidate center set, and $\Delta(a,b)$ denotes the connection cost between two points $a$ and $b$. For $k$-MedP and $k$-MedO, we have $\cC = \cF$ and $\Delta(a,b) = d(a,b)$; for $k$-MeaP and $k$-MeaO, we have $\cC = \cX$ and $\Delta(a,b) = d^2(a,b)$. Then, the penalty-version  can be formulated as
$$
\min_{S \subseteq \cC, P \subseteq \cX} \sum\limits_{x \in \cX \setminus P} \min_{s \in S} \Delta(s,x) + \sum\limits_{x \in P} p_x,
$$
and the outlier-version can be formulated as
$$
\min_{S \subseteq \cC, P \subseteq \cX: |P|\le z} \sum\limits_{x \in \cX \setminus P} \min_{s \in S} \Delta(s,x).
$$

Considering $k$-MeaP and $k$-MedP, we assume that $S$ is a set of $k$ centers. It is obvious that the optimal penalized point set with respect to $S$ is $P = \{x  \in \mathcal{X} | p_x \le  \min_{s \in S} d(s, x) \}$ for $k$-MedP  and $P =\{x  \in \mathcal{X}| p_x \le  \min_{s \in S} d^2(s, x) \}$ for $k$-MeaP, implying that $S$ determines the corresponding $k$ clusters $ N (s) := \{x \in \mathcal{X}\setminus P |  s_x = s\}$ for all $s \in S$, where $ s_x$ denotes the closest center in $S$ to $x  \in \mathcal{X} \setminus P $, i.e., $s_x :=  \argmin_{s \in S} d (s,x)$. Thus, we also call $S$ a feasible solution for $k$-MedP and $k$-MeaP.

Given a center set $S$ and a subset $R \subseteq \cX$, we suppose that $\cX \setminus R = \{ x_1,x_2,\dots,$ $x_{|\cX \setminus R|} \}$ subject to $d(s_{x_1},x_1) \ge d(s_{x_2},x_2) \ge \dots \ge d(s_{x_{|\cX \setminus R|}},x_{|\cX \setminus R|})$. Let $\outlier(S,R)$ $:= \{x_1, x_2, \dots, x_z\}$ if $|\cX \setminus R| \ge z$, otherwise, $\outlier(S,R) := \cX \setminus R$. We simplify $\outlier(S,\emptyset)$ to $\outlier(S)$.
For $k$-MedO and $k$-MeaO, it is obvious that the optimal outlier set with respect to $S$ is $\outlier(S)$, implying that the set $S$ can be seen as a feasible solution. We also use $(S,P)$ to denote a solution (not necessarily feasible) in which the center set is $S$ and  the outlier set is $P$ for $k$-MedO and $k$-MeaO.

\subsection{Some technical lemmas}\label{sec:tech-lem}
Given a data subset $D \subseteq  \mathcal{X}$ and a point  $ c  \in \cC$,
we define $ \Delta(c,D) := \sum_{x \in D} \Delta(c,x) $. Let $\cent_{\cC}(D)$ be a center point in $\cC$ that optimizes the objective of the $k$-means/$k$-median  problem, i.e.,  $\cent_{\cC}(D) := \argmin_{c \in \cC} \Delta(c,D)$. We remark that the notation $\argmin$ ($\argmax$) denotes an arbitrary element that minimizes (maximizes) the objective.  From the well-known centroid lemma (\citealt{kmnpsw}), we get $\cent_{\cC}(D) = \cent(D)$ for $k$-means, where $  {\rm cent} (D)$ is the centroid of $D$, that is defined as follows.
\begin{definition}[Centroid]\label{defcentroid}
Given a set $D \subseteq \bR^d$, we call the point $\sum_{x \in D} x / |D|$ denoted by ${\rm cent} (D)$  the centroid of $D$.
\end{definition}

\begin{lemma}[Centroid Lemma \citep{kmnpsw}] \label{lem2.1}
For any data subset $D \subseteq \mathcal{X} $ and a point $ c \in \mathbb{R}^{d}$,
we have $d^2 (c,D) =
d^2({\rm cent} (D), D) + |D| d^2 ( {\rm cent} (D), c)$.
\end{lemma}

So, the candidate center points of a $k$-means problem are the centroid points for all subsets of $\mathcal{X} $.
Note that the total amount of these candidate center points is $2^{|\mathcal{X}|}-1 $.
To cut down this exponential magnitude, \cite{mat} introduces the concept of approximate centroid set  shown in the following definition.
\begin{definition}\label{defi1}
A set $\cC' \subseteq  \mathbb{R}^d$ is an
$\varepsilon$-approximate centroid set for $ \mathcal{X} \subseteq \mathbb{R}^d $
if for any set $ D \subseteq \mathcal{X} $, we have
$ \min_{c \in \cC' }  d^2 (c,D) \le
( 1+ \varepsilon ) \min_{c \in \mathbb{R}^d }  d^2 (c,D)$.
\end{definition}

The following lemma shows the important observation  that
a  polynomial size $ \hat \epsilon$-approximate centroid set for $ \mathcal{X} $ can be found in polynomial time.
In the remainder of this paper, we restrict that the candidate center set of $k$-MeaP/$k$-MeaO is the $ \hat \varepsilon$-approximate centroid set $\cC'$,
by utilizing this observation.
\begin{lemma}[\cite{mat}] \label{lem2.2}
Given an $n$-point set $\cX$ and a real number $\veps>0$, an $\veps$-approximate centroid set for $\cX$, of size $O\left(n\veps^{-d}\log(1/\veps)\right)$, can be computed in time $O\left(n\log n+n\veps^{-d}\log(1/\veps)\right)$.
\end{lemma}

For the $k$-median problem, we do not need the approximate center set, since the candidate centers are in the finite set $\cF$.

\section{Local search approximation algorithms  for  $k$-MedP and $k$-MeaP}\label{sec:penalty}
Let  $\rho$ be a fixed integer.
For any feasible solution $S$, $ A \subseteq S $ and $ B \subseteq \mathcal{C}\setminus S$ with $ |A| =|B| \le \rho $, we define the so-called multi-swap operation ${\rm swap}(A$, $B)$
such that all centers in $A$ are dropped from $S$ and  all centers in $B$ are added to $S$.

We further denote the connection cost of the point $x \in \cX$ by $\cost_c(x)$, i.e., $\cost_c(x):=\Delta(s_x,x)$, and denote by ${\rm cost}_c $, ${\rm cost}_p $, and ${\rm cost}(S)$
the following expressions
${\rm cost}_c  : = \sum_{x \in  \mathcal{X} \setminus P} {\rm cost}_c (x)$;
${\rm cost}_p : = \sum_{x \in  P} p_x $;
${\rm cost}  (S):={\rm cost}_c  + {\rm cost}_p$, where $P$ is the optimal penalized point set with respect to $S$.

Now we are ready to present our multi-swap local search algorithm.
\vspace{0.5cm}
\begin{breakablealgorithm}
\renewcommand{\algorithmicrequire}{\textbf{Input:}}
\renewcommand{\algorithmicensure}{\textbf{Output:}}
\caption{The multi-swap local search algorithm: LS-Multi-Swap($\cX,C,k,\{p_j\}_{j \in \cX},\rho$)}
\begin{algorithmic}[1]\label{alg-multi-swap}
    \REQUIRE data set $\cX$, candidate center set $C$, penalty cost $p_j$ for all $j \in \cX$, positive integers $k$ and $\rho \le k$.
    \ENSURE center set $S \subseteq C$.
    \STATE Arbitrarily choose a $k$-center subset $S$ from $ C$.
    \STATE Compute
        $(A,B): = \arg\min_{A\subseteq S, B \subseteq C\setminus S, |A|=|B| \le \rho}
        {\rm cost} (S\setminus A \cup B).$
    \WHILE {${\rm cost} (S\setminus A \cup B) <  {\rm cost} (S)$}
        \STATE Set $S:= S\setminus A \cup B$.
        \STATE Compute
            $(A,B): = \arg\min_{A\subseteq S, B \subseteq C\setminus S, |A|=|B| \le \rho}
            {\rm cost} (S\setminus A \cup B).$
    \ENDWHILE
    \RETURN $S$
\end{algorithmic}
\end{breakablealgorithm}
\vspace{0.5cm}

For $k$-MedP, we run LS-Multi-Swap($\cX,\cF,k,\{p_j\}_{j \in \cX},\rho$); for $k$-MeaP, we first call the algorithm of \cite{mmsw} to construct an $ {\hat \varepsilon}$-approximate centroid set $ \mathcal{\cC}' \subseteq \cX$, then run LS-Multi-Swap($\cX,\cC',k,\{p_j\}_{j \in \cX},\rho$). The values of $\rho$ and $\hat{\veps}$ will be determined in our analysis of the algorithm.

\subsection{The analysis}
Let  $S^*$ be  a global optimal solution  with the penalized set
$ P^* =\{x  \in \mathcal{X} | p_x \le \min_{s^* \in S} \Delta(x, s^*) \}  $.
Similar to the feasible solution $S$, we introduce the corresponding notations
$ s^*_x$, $ N^* (s^*)$, $ {\rm cost}^*_c (x)$, ${\rm cost}_c^*$, ${\rm cost}_p^*$
and ${\rm cost}  (S^*)$.

We use the standard analysis  for a local search algorithm, in which some swap operations between $S$ and $S^*$ are constructed, and then each point is reassigned to a center in the new solution. In the cluster form analysis, we try to bound the new cost for a set of points, rather than bounding the cost of each point individually and independently. To this end, we introduce the \emph{adapted cluster} as follows.
$$
N^*_q (s^*): = N^* (s^*) \setminus  P, \qquad \forall  s^* \in S^*.
$$

With the adapted cluster, we set $\tilde S^* :=  \{ \cent_{\cC}(N^*_q (s^*) )|s^* \in S^*  \} $.
We introduce
a mapping $ \phi: \tilde S^* \rightarrow S$ and map each point $ c \in\tilde S^*$ to $ \phi(c) : = \arg \min_{s \in S} d (c, s) $. We say that the center $\phi(\cent_{\cC}(N^*_q (s^*))$ {\emph{captures}} $s^*$.
Considering one of all constructed swap operations, we will reassign some points to a center determined by the mapping $\phi$ (for instance, reassign the point $x$ to $\phi( {\rm cent}  (N^*_q (s^*_x) ))$. The details will be stated later).

Combining all swap operations, the sum of the costs of these points appears in the right hand side of the inequality which is derived from the local optimality of $S$. For $k$-MeaP, we can bound this sum by the connection costs of $S$ and $S^*$, see Lemma \ref{lem-ub-cluster form}.
Note that all these points are not outliers in both $S$ and $S^*$. This is the reason why we need to use the adapted cluster rather than the cluster $N^*(s^*)$ which was used in the analysis for $k$-means \citep{gklmv}.

In the proof of Lemma \ref{lem-ub-cluster form}, we divide the set $\cX \setminus (P \cup P^*)$ into some adapted clusters with respect to all $s^* \in S^*$, and apply the Centroid Lemma to each adapted cluster. Afterwards we bound the square of distances between a centroid $c$ of the adapted cluster and its mapped point $\phi(c)$. This explains why the domain of the mapping $\phi$ is the set of centroids of adapted clusters.

\begin{lemma} \label{lem-ub-cluster form}
Let $ S$ and $S^*$ be a local optimal solution and a global optimal solution of
$k$-MeaP, respectively. Then,
\bea
\sum\limits_{x \in  \mathcal{X}  \setminus (P \cup P^*) }   d^2 (\phi( {\rm cent}  (N^*_q (s^*_x) )), x)
& \le &
\sum\limits_{x \in  \mathcal{X}  \setminus (P \cup P^*) }
\left(
2 {\rm cost}^*_c (x) +  {\rm cost}_c (x)
\right) +
\nn\\
&&
2 \sqrt{  \sum\limits_{x \in  \mathcal{X}  \setminus (P \cup P^*) }  {\rm cost}^*_c (x) }
\cdot  \sqrt{ \sum\limits_{x \in  \mathcal{X}  \setminus (P \cup P^*) }  {\rm cost}_c (x) }.
\eea
\end{lemma}

\proof
With the Cauchy-Schwarz inequality,
we obtain
\bea
& & \sum\limits_{s^* \in S^*}  \sum\limits_{x \in  N^*_q (s^* ) }  d  (x, \cent_{\cC}(N^*_q (s^*) )  ) \cdot   d   (x,   s_x )
\nn \\
& \le &  \sqrt{ \sum\limits_{s^* \in S^*}  \sum\limits_{x \in  N^*_q (s^* ) }  d^2  (x, \cent_{\cC}(N^*_q (s^*) )  ) }
\cdot  \sqrt{ \sum\limits_{s^* \in S^*}  \sum\limits_{x \in  N^*_q (s^* ) }   d^2   (x,   s_x ) }.
\label{eq2-lem-ub-cluster form}
\eea

Lemma \ref{lem2.1}  and the definition of $\phi( \cdot )$ then yield
\bea
& &  \sum\limits_{x \in  \mathcal{X}  \setminus (P \cup P^*) }   d^2 (\phi( \cent_{\cC}( N^*_q (s^*_x) )), x)
\nn \\
& =  &
\sum\limits_{s^* \in S^*}   \sum\limits_{x \in  N^*_q (s^*) }   d^2 (\phi( \cent_{\cC}(  N^*_q (s^*) )), x)
\nn \\
& =  &
\sum\limits_{s^* \in S^*}  \left[
d^2 (\cent_{\cC}(N^*_q (s^*) ), N^*_q (s^*)  )
+ | N^*_q (s^*) | \cdot   d^2 (\cent_{\cC}(  N^*_q (s^*) ),  \phi (\cent_{\cC}(  N^*_q (s^*) )) )
\right]
\nn \\
& =  &
\sum\limits_{s^* \in S^*}
d^2 (\cent_{\cC}(  N^*_q (s^*) ), N^*_q (s^*)  )
+ \sum\limits_{s^* \in S^*}  \sum\limits_{x \in  N^*_q (s^* ) }    d^2 (\cent_{\cC}(  N^*_q (s^*) ),  \phi (\cent_{\cC}(  N^*_q (s^*) )) )
\nn  \\
& \le &
\sum\limits_{s^* \in S^*}
d^2 (\cent_{\cC}(  N^*_q (s^*) ), N^*_q (s^*)  )
+ \sum\limits_{s^* \in S^*}  \sum\limits_{x \in  N^*_q (s^* ) }    d^2 (\cent_{\cC}(  N^*_q (s^*) ),   s_x ).
\label{eq4-lem-ub-cluster form}
\eea
Using the triangle inequality for $ d  (\cdot, \cdot)$, we obtain
\bea
& &
\sum\limits_{s^* \in S^*}  \sum\limits_{x \in  N^*_q (s^* ) }    d^2 (\cent_{\cC}(  N^*_q (s^*) ),   s_x )
\nn \\
& \le &
\sum\limits_{s^* \in S^*}  \sum\limits_{x \in  N^*_q (s^* ) }
\left(
d  (x, \cent_{\cC}(  N^*_q (s^*) )  )
+ d  (x,   s_x )
\right)^2
\nn \\
& =&
\sum\limits_{s^* \in S^*}
d^2 (\cent_{\cC}(N^*_q (s^*) ), N^*_q (s^*)  )
+  \sum\limits_{s^* \in S^*}  \sum\limits_{x \in  N^*_q (s^* ) }  d^2  (x,   s_x )
\nn \\
& &
+\ 2 \sum\limits_{s^* \in S^*}  \sum\limits_{x \in  N^*_q (s^* ) }  d  (x, \cent_{\cC}(  N^*_q (s^*) )  ) \cdot   d   (x,   s_x ).
\label{eq5-lem-ub-cluster form}
\eea
Integrating  \reff{eq2-lem-ub-cluster form}-\reff{eq5-lem-ub-cluster form} and using the definition of $\cent_{\cC}( \cdot)$ then gives
\bea
& &
\sum\limits_{x \in  \mathcal{X}  \setminus (P \cup P^*) }   d^2 (\phi( \cent_{\cC}(  N^*_q (s^*_x) )), x)
\nn \\
& \le &
2 \sum\limits_{s^* \in S^*}
d^2 (\cent_{\cC}(  N^*_q (s^*) ), N^*_q (s^*)  )
+  \sum\limits_{s^* \in S^*}  \sum\limits_{x \in  N^*_q (s^* ) }  d^2  (x,   s_x )
\nn \\
& &
+\ 2 \sqrt{ \sum\limits_{s^* \in S^*}  \sum\limits_{x \in  N^*_q (s^* ) }  d^2  (x, \cent_{\cC}(  N^*_q (s^*) )  ) }
\cdot  \sqrt{ \sum\limits_{s^* \in S^*}  \sum\limits_{x \in  N^*_q (s^* ) }   d^2   (x,   s_x ) }
\nn  \\
& \le &
2 \sum\limits_{s^* \in S^*}
d^2 (s^*, N^*_q (s^*)  )
+  \sum\limits_{s^* \in S^*}  \sum\limits_{x \in  N^*_q (s^* ) }  d^2  (x,   s_x )
\nn \\
& &
+\ 2 \sqrt{ \sum\limits_{s^* \in S^*}  \sum\limits_{x \in  N^*_q (s^* ) }  d^2  (x, s^* ) }
\cdot  \sqrt{ \sum\limits_{s^* \in S^*}  \sum\limits_{x \in  N^*_q (s^* ) }   d^2   (x,   s_x ) }
\nn  \\
& = &
\sum\limits_{x \in  \mathcal{X}  \setminus (P \cup P^*) }
\left(
2 {\rm cost}^*_c (x) +  {\rm cost}_c (x)
\right)
+ 2 \sqrt{  \sum\limits_{x \in  \mathcal{X}  \setminus (P \cup P^*) }  {\rm cost}^*_c (x) }
\cdot  \sqrt{ \sum\limits_{x \in  \mathcal{X}  \setminus (P \cup P^*) }  {\rm cost}_c (x) }.
\nn
\eea
Note that $\cent_{\cC}(\cdot) = \cent(\cdot)$ for $k$-MeaP. So we complete the proof.
\endproof

Consider now $ \phi (\tilde S^*  )$, i.e., the image set of $\tilde{S^*}$ under $\phi$.
We list all elements of $\phi (\tilde S^*)$ as $ \phi (\tilde S^*) = \{s_1, ..., s_m \}$
where  $ m:= | \phi (\tilde S^*) |$.
For each  $  l \in \{1, ..., m \} $,
let
$ S_l : = \{ s_l\}$
and
$ S^*_l : =   \{ s^* \in S^* | \phi (\cent_{\cC}(N^*_q ($ $s^*) ))= s_l  \}   $.
Thus, $S^*$ is partitioned into  $S^*_1,S^*_2, ..., S^*_m $.
Noting that $|S|= |S^*|=k $,
we can enlarge each $S_l$
such that $S_1,S_2, ..., S_m $ is a partition of $S$
with $ |S_l| = |S^*_l | $ for each $ l \in \{ 1,2,...,m \}$.

We will construct a swap operation  between the points in $S_l$ and $S^*_l$ for each pair $(S_l,S^*_l)$.
Before doing this, we note that a center $s^* \in S^*$ need not belong to the candidate center set $\cC'$ for $k$-MeaP. Thus, we introduce a center $ {\hat s^*} \in \mathcal{C} $
associated with each $s^* \in S^*$ to ensure that the swap operation involved in $s^*$
can be implemented in Algorithm \ref{alg-multi-swap}.
For each $s^* \in S^*$, let ${\hat s^*}: = \arg\min_{c \in \cC'} d (c,N^* (s^*))$.
Combined with Definition \ref{defi1},
we have (see \citealt{zhwxz})
\bea \label{ieq-approx-cent}
\sum\limits_{x \in N^* (s^*)} d^2 ({\hat s^*}, x)
& = &  d^2 ({\hat s^*},N^* (s^*))
  =    \min\limits_{c \in \mathcal{C} } d^2 (c,N^* (s^*))
\nn \\
& \le &
(1+ {\hat \varepsilon})  \min\limits_{c \in \mathbb{R}^d } d^2 (c,N^* (s^*))
  =
(1+ {\hat \varepsilon})  d^2 (s^*,N^* (s^*))
\nn \\
& = & (1+ {\hat \varepsilon}) \sum\limits_{x \in N^* (s^*)} d^2 (s^*,x).
\eea

The algorithm allows at most $\rho$ points to be swapped. To satisfy this condition, we consider the following two cases to construct swap operations (cf. Figure \ref{fig:swap1} for $\rho=3$).

\begin{figure}[t]
\centering
  \subfigure[$|S_l| \le \rho$]{
    \label{fig:multiswap} 
    \includegraphics[height=3.0cm]{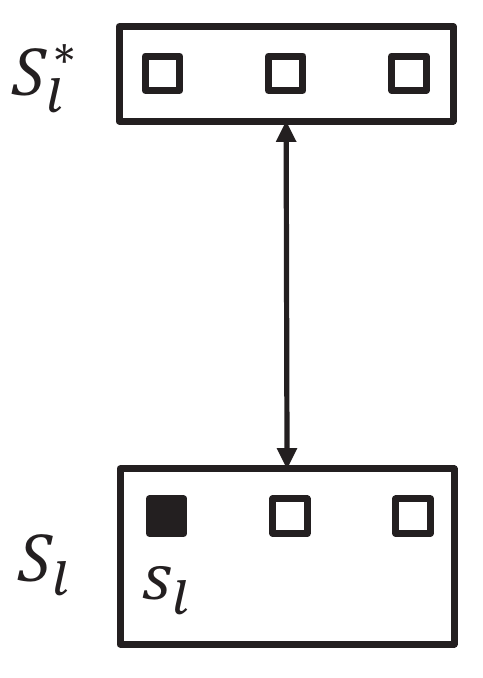}}
  \hspace{2.5cm}
  \subfigure[$|S_l| > \rho$]{
    \label{fig:singleswap1} 
    \includegraphics[height=3.0cm]{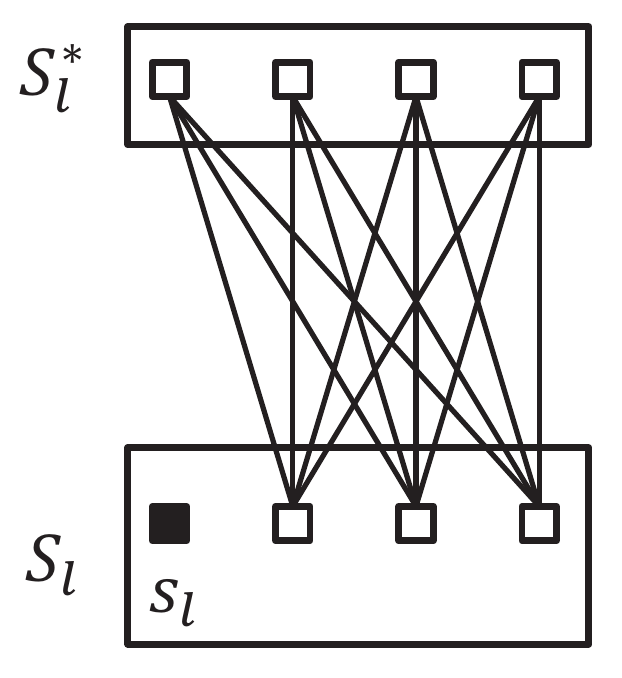}}
\caption{Two cases for constructing the swap operations between $S_l$ and $S_l^*$ for $\rho=3$. The solid squares belong to $\phi(\tilde{S}^*)$. }
\label{fig:swap1}
\end{figure}

\begin{description}
\item[Case 1](cf. Figure \ref{fig:multiswap}). For each $l$ with $ |S_l| = |S^*_l| \le \rho$,
we consider the pair $(S_l, S^*_l)$.
Let  ${\hat S^*_l}: = \{ \hat s^* |   s^* \in S^*_l \} $.
W.l.o.g., we assume that ${\hat S^*_l} \subseteq  \cX \setminus S $.
For $k$-MedP, we consider the swap$(S_l,S^*_l)$; for $k$-MeaP, we consider the swap$(S_l, {\hat S^*_l}) $.
Utilizing  these swap operations, we obtain the following result.
\begin{lemma}\label{lem-penalty-case1}
If $ |S_l| = |S^*_l| \le \rho$, then we have
\bea
0 & \le &
\sum\limits_{s \in  S_l }  \sum\limits_{x \in  N(s)\cap P^*} (p_x - {\rm cost}_c (x))+
\sum\limits_{s \in  S_l }  \sum\limits_{x \in  N(s) \setminus  P^* }
\left(  d(\phi( \cent_{\cC}(N^*_q (s^*_x) )), x)  - {\rm cost}_c (x) \right)+
\nn\\
&&
\sum\limits_{s^* \in  S_l^* } \sum\limits_{x \in   N^* (s^*) \setminus  P  } ({\rm cost}_c^* (x)  - {\rm cost}_c (x))+
\sum\limits_{s^* \in  S_l^* } \sum\limits_{x \in   N^* (s^*)  \cap  P  } ( {\rm cost}_c^* (x)  - p_x). \label{ieq1-lem-penalty-case1}
\eea
for $k$-MedP, and
\bea
0 & \le &
\sum\limits_{s \in  S_l }  \sum\limits_{x \in  N(s)\cap P^*} (p_x - {\rm cost}_c (x))+
\sum\limits_{s \in  S_l }  \sum\limits_{x \in  N(s) \setminus  P^* }
\left(  d^2 (\phi( \cent_{\cC}(N^*_q (s^*_x) )), x)  - {\rm cost}_c (x) \right)+
\nn\\
&&
\sum\limits_{s^* \in  S_l^* } \sum\limits_{x \in   N^* (s^*) \setminus  P  } ((1+ {\hat \varepsilon})  {\rm cost}_c^* (x)  - {\rm cost}_c (x))+
\sum\limits_{s^* \in  S_l^* } \sum\limits_{x \in   N^* (s^*)  \cap  P  } ( (1+ {\hat \varepsilon}) {\rm cost}_c^* (x)  - p_x). \label{ieq2-lem-penalty-case1}
\eea
for $k$-MeaP.
\end{lemma}

\item[Case 2](cf. Figure \ref{fig:singleswap1}). For each $l$ with $ |S_l| = |S^*_l| = m_l > \rho$,
we consider $ (m_l-1)m_l$ pairs $(s, s^*)$ with
$s \in S_l \backslash \{ s_l\} $ and $ s^* \in S^*_l$.
For $k$-MedP, we consider the swap$(s, s^*)$; for $k$-MeaP, we consider the swap$(s, {\hat s^*})$.
Utilizing  these swap operations, we obtain the following result.
\begin{lemma}\label{lem-penalty-case2}
For any $s \in S_l \backslash \{ s_l\} $ and $ s^* \in S^*_l$, we have
\bea
0 & \le &
\sum\limits_{x \in  N(s)\cap P^*} (p_x - {\rm cost}_c (x))+
\sum\limits_{x \in  N(s) \setminus  P^* }
\left(  d(\phi( \cent_{\cC}(N^*_q (s^*_x) )), x)  - {\rm cost}_c (x) \right)+
\nn \\
& &
\sum\limits_{x \in   N^* (s^*) \setminus  P  } ({\rm cost}_c^* (x)  - {\rm cost}_c (x))+
\sum\limits_{x \in   N^* (s^*)  \cap  P  } ({\rm cost}_c^* (x)  - p_x)
\label{ieq1-lem-penalty-case2}
\eea
for $k$-MedP, and
\bea
0 & \le &
\sum\limits_{x \in  N(s)\cap P^*} (p_x - {\rm cost}_c (x)) +
\sum\limits_{x \in  N(s) \setminus  P^* }
\left(  d^2 (\phi( \cent_{\cC}(N^*_q (s^*_x) )), x)  - {\rm cost}_c (x) \right) +
\nn \\
&&
\sum\limits_{x \in   N^* (s^*) \setminus  P  } ((1+ {\hat \varepsilon})  {\rm cost}_c^* (x)  - {\rm cost}_c (x)) +
\sum\limits_{x \in   N^* (s^*)  \cap  P  } ( (1+ {\hat \varepsilon}) {\rm cost}_c^* (x)  - p_x)
\label{ieq2-lem-penalty-case2}
\eea
for $k$-MeaP.
\end{lemma}
\end{description}

Lemma \ref{lem-penalty-case1} shows a relationship between the sets $S_l$ and $S^*_l$, while Lemma \ref{lem-penalty-case2} shows a relationship between two points in $S_l$ and $S^*_l$ respectively. We remark that Lemma \ref{lem-penalty-case2} holds for all pairs $(S_l,S^*_l)$ (no matter whether $|S_l| > \rho$). This is useful for the analysis of the algorithm for $k$-MedO/$k$-MeaO in Section \ref{sec:outlier}.

\begin{proof}[Proof of Lemma \ref{lem-penalty-case1}]
We only prove it for $k$-MeaP. The proof for $k$-MedP is similar.
After the operation swap$(S_l, {\hat S^*_l})$, we penalize all points in $N(s) \cap P^*$ for all $s \in S_l$, reassign each point $x \in N^*(s^*)$ to $\hat{s}^*$ for all $s^* \in S^*_l$, and reassign  each point $x \in  N(s)  \setminus \left( \bigcup_{s^* \in S^*_l} N^* (s^*) \cup P^*\right)$ to $\phi(\cent_{\cC}(N^*_q (s^*_x) ))$ for all $s \in S_l$ ($s^*_x \notin S^*_l$ implies $\phi(\cent_{\cC}(N^*_q (s^*_x) )) \notin S_l$).
Since the  operation swap$(S_l, {\hat S^*_l}) $ does not improve the local optimal solution $S$,
we have
\bea
0 & \le &
{\rm cost}(S \setminus \{ S_l \} \cup \{ \hat S^*_l \} ) - {\rm cost}(S)
\nn \\
& \le &
\sum\limits_{s \in S_l } \sum\limits_{x \in  N(s)\cap P^*} (p_x - {\rm cost}_c (x)) +
\nn\\
&&
\sum\limits_{s \in  S_l } \sum\limits_{x \in  N(s)  \setminus \left( \bigcup_{s^* \in S^*_l} N^* (s^*) \cup P^*\right)   }
\left( d^2 (\phi( \cent_{\cC}(N^*_q (s^*_x) )), x)   - {\rm cost}_c (x) \right)+
\nn \\
& &
\sum\limits_{s^* \in  S^*_l }  \sum\limits_{x \in   N^* (s^*) \setminus  P  } ( d^2 ({\hat s^*}, x)  - {\rm cost}_c (x))
+ \sum\limits_{s^* \in  S^*_l }  \sum\limits_{x \in   N^* (s^*)  \cap  P  } ( d^2 ({\hat s^*}, x)  - p_x). \nn\
\eea
Combining this with $\sum_{x \in  N^* (s^*)} = \sum_{x \in   N^* (s^*) \setminus P} + \sum_{x \in   N^* (s^*)  \cap  P}$ and inequality \reff{ieq-approx-cent} completes the proof.
\end{proof}

\begin{proof}[Proof of Lemma \ref{lem-penalty-case2}]
We again only prove it for $k$-MeaP, and the proof for $k$-MedP is again similar.
Recall the definition of ${\hat s^*} $.
W.l.o.g., we assume that ${\hat s^*} \notin S$. It follows from $s \neq s_l$ and $\phi(\cent_{\cC}(N^*_q(s^*)))=s_l$ that
$ \phi( \cent_{\cC}(N^*_q (s^*_x) ))$ $\ne  s $ when $ x \in  N (s) \setminus ( N^* (s^*) \cup P^*) $.
Since the  operation swap$(s, {\hat s^*})$ does not improve the current solution $S$,
we have
\bea
0 & \le &
{\rm cost}(S \setminus \{ s \} \cup \{ \hat s^* \} ) - {\rm cost}(S)
\nn \\
& \le &
\sum\limits_{x \in  N (s)\cap P^*} (p_x - {\rm cost}_c (x))
+ \sum\limits_{x \in  N (s) \setminus ( N^* (s^*) \cup P^*) }
\left( d^2 (\phi( \cent_{\cC}(N^*_q (s^*_x) )), x)   - {\rm cost}_c (x) \right)
\nn \\
& &
+ \sum\limits_{x \in   N^* (s^*) \setminus  P  } ( d^2 ({\hat s^*}, x)  - {\rm cost}_c (x))
+ \sum\limits_{x \in   N^* (s^*)  \cap  P  } ( d^2 ({\hat s^*}, x)  - p_x)
\nn \\
& \le &
\sum\limits_{x \in  N(s)\cap P^*} (p_x - {\rm cost}_c (x))
+ \sum\limits_{x \in  N(s) \setminus  P^* }
\left(  d^2 (\phi( \cent_{\cC}(N^*_q (s^*_x) )), x)  - {\rm cost}_c (x) \right)
\nn \\
& &
+ \sum\limits_{x \in   N^* (s^*) \setminus  P  } ((1+ {\hat \varepsilon})  {\rm cost}_c^* (x)  - {\rm cost}_c (x))
+ \sum\limits_{x \in   N^* (s^*)  \cap  P  } ( (1+ {\hat \varepsilon}) {\rm cost}_c^* (x)  - p_x).\nn
\eea
This completes the proof.
\end{proof}

Combining Lemmas \ref{lem-penalty-case1} and \ref{lem-penalty-case2}, we estimate the cost of $S$ for $k$-MedP and $k$-MeaP in the following two theorems respectively.
\begin{theorem} \label{thm-multi-swap-kmedp}
{\rm LS-Multi-Swap($\cX,\cF,k,\{p_j\}_{j \in \cX},\rho$)} for $k$-MedP
produces  a local optimal solution $S$ satisfying
${\rm cost}_c +  {\rm cost}_p \le  (3+ 2/\rho)  {\rm cost}_c^*
+  ( 1+  1/\rho )  {\rm cost}_p^*$.
\end{theorem}
\begin{theorem} \label{thm-multi-swap-kmeap}
Let $\cC'$ be an $\hat{\veps}$-approximate centroid set for $\cX$. {\rm LS-Multi-Swap($\cX,$ $\cC',k,\{p_j\}_{j \in \cX},\rho$)} for $k$-MeaP
produces  a local optimal solution $S$ satisfying
${\rm cost}_c +  {\rm cost}_p \le  \left( 3 + 2/\rho + {\hat \varepsilon} \right)^2 {\rm cost}_c^* +  \left(3+ 2/\rho + {\hat \varepsilon} \right) \left(1+ 1/\rho \right)
 {\rm cost}_p^*$.
\end{theorem}

\proof[Proof of Theorem \ref{thm-multi-swap-kmedp}.]
Note that $ m_l/(m_l-1) \le (\rho+ 1)/\rho $ and $d(\phi(c), x) \ge {\rm cost}_c (x)$ for any $c \in \tilde{S}^*$ and any $x \in \cX$.
Summing the inequality \reff{ieq1-lem-penalty-case1} with weight $1$
and inequality \reff{ieq1-lem-penalty-case2} with weight $1/(m_l -1)$
over all constructed swap operations, we have
\bea
0 & \le &
\left(1+ \frac{1}{\rho} \right)
\sum\limits_{s \in S} \sum\limits_{x \in  N(s)\cap P^*} (p_x - {\rm cost}_c (x))
\nn \\
& &
+ \left(1+ \frac{1}{\rho} \right)
\sum\limits_{s \in S}  \sum\limits_{x \in  N(s) \setminus  P^* }
\left(  d(\phi( \cent_{\cC}(N^*_q (s^*_x) )), x)   - {\rm cost}_c (x) \right)
\nn \\
& &
+   \sum\limits_{s^* \in S^*}  \sum\limits_{x \in   N^* (s^*) \setminus  P  }
({\rm cost}_c^* (x)  - {\rm cost}_c (x))
\nn \\
& &
+ \sum\limits_{s^* \in S^*}   \sum\limits_{x \in   N^* (s^*)  \cap  P  }
({\rm cost}_c^* (x)  - p_x).
\label{ieq1-thm-penalty-kmedp}
\eea
The triangle inequality and the definition of $\phi(\cdot)$ imply that
\bea
&& d(\phi(\cent_{\cC}(N^*_q (s^*_x))), x) \nn\\
& \le &
d(\phi(\cent_{\cC}(N^*_q (s^*_x))),\cent_{\cC}(N^*_q (s^*_x))) + d(\cent_{\cC}(N^*_q (s^*_x)),x) \nn\\
& \le &
d(s_x,\cent_{\cC}(N^*_q (s^*_x))) + d(\cent_{\cC}(N^*_q (s^*_x)),x) \nn\\
& \le &
d(s_x,x) + d(\cent_{\cC}(N^*_q (s^*_x)),x) + d(\cent_{\cC}(N^*_q (s^*_x)),x) \nn\\
& = &
\cost_c(x) + 2d(\cent_{\cC}(N^*_q (s^*_x)),x).
\label{ieq2-thm-penalty-kmedp}
\eea
Combining inequalities \reff{ieq1-thm-penalty-kmedp} and \reff{ieq2-thm-penalty-kmedp}, we obtain
\bea
0 & \le &
\left(1+ \frac{1}{\rho} \right)
\sum\limits_{s \in S} \sum\limits_{x \in  N(s)\cap P^*} (p_x - {\rm cost}_c (x))
+ \left(1+ \frac{1}{\rho} \right)
\sum\limits_{s \in S}  \sum\limits_{x \in  N(s) \setminus  P^* }
2d(\cent_{\cC}(N^*_q (s^*_x)),x)
\nn \\
& &
+\sum\limits_{s^* \in S^*}  \sum\limits_{x \in   N^* (s^*) \setminus  P  }
({\rm cost}_c^* (x)  - {\rm cost}_c (x))
+ \sum\limits_{s^* \in S^*}   \sum\limits_{x \in   N^* (s^*)  \cap  P  }
({\rm cost}_c^* (x)  - p_x)
\nn\\
& \le &
\left(1+ \frac{1}{\rho} \right)
\sum\limits_{x \in  P^*} p_x
+ \left(1+ \frac{1}{\rho} \right)
\sum\limits_{s^* \in S^*}  \sum\limits_{x \in  N_q^*(s^*)}
2d(\cent_{\cC}(N^*_q (s^*_x)),x)
\nn \\
& &
+ \sum\limits_{x \in \cX \setminus P^*} {\rm cost}_c^* (x)
- \sum\limits_{x \in \cX \setminus P} {\rm cost}_c (x)
- \sum\limits_{x \in P} p_x \nn\\
& = &
\left(1+ \frac{1}{\rho} \right) \cost_p^*
+ \left(1+ \frac{1}{\rho} \right)
\sum\limits_{s^* \in S^*}  \sum\limits_{x \in  N_q^*(s^*)}
2d(\cent_{\cC}(N^*_q (s^*_x)),x)
\nn\\
&&
+\ \cost_c^* - \cost_c - \cost_p.  \label{ieq3-thm-penalty-kmedp}
\eea
From the definitions of $N^*_q(\cdot)$ and $\cent_{\cC}(\cdot)$, we get that
\bea \label{ieq4-thm-penalty-kmedp}
\sum\limits_{s^* \in S^*}  \sum\limits_{x \in  N_q^*(s^*)}
2d(\cent_{\cC}(N^*_q (s^*_x)),x) \le \sum\limits_{s^* \in S^*}  \sum\limits_{x \in  N_q^*(s^*)} 2d(s^*_x,x) \le 2\cost_c^*.
\eea
Finally, we complete the proof by combining  inequalities \reff{ieq3-thm-penalty-kmedp}-\reff{ieq4-thm-penalty-kmedp} for $\rho = 2 / \veps$.
\endproof

\proof[Proof of Theorem \ref{thm-multi-swap-kmeap}.]
Similar to the proof of Theorem \ref{thm-multi-swap-kmedp}, we obtain by summing  inequality \reff{ieq2-lem-penalty-case1} with weight $1$ and inequality \reff{ieq2-lem-penalty-case2} with weight $1/(m_l -1)$ over all constructed swap operations that
\bea
0 & \le &
\left(1+ \frac{1}{\rho} \right)
\sum\limits_{s \in S} \sum\limits_{x \in  N(s)\cap P^*} (p_x - {\rm cost}_c (x))
\nn \\
& &
+ \left(1+ \frac{1}{\rho} \right)
\sum\limits_{s \in S}  \sum\limits_{x \in  N(s) \setminus  P^* }
\left(  d^2 (\phi( \cent_{\cC}(N^*_q (s^*_x) )), x)   - {\rm cost}_c (x) \right)
\nn \\
& &
+   \sum\limits_{s^* \in S^*}  \sum\limits_{x \in   N^* (s^*) \setminus  P  }
((1+ {\hat \varepsilon}) {\rm cost}_c^* (x)  - {\rm cost}_c (x))
\nn \\
& &
+ \sum\limits_{s^* \in S^*}   \sum\limits_{x \in   N^* (s^*)  \cap  P  }
( (1+ {\hat \varepsilon}) {\rm cost}_c^* (x)  - p_x).
\label{ieq1-thm-penalty-kmeap}
\eea
Because of $\sum_{s \in S} \sum_{x \in N(s) \setminus P^*} = \sum_{x \in \cX \setminus (P \cup P^*)}$ and Lemma \ref{lem-ub-cluster form},  the RHS of \reff{ieq1-thm-penalty-kmeap} is not larger than
\bea
&   &
\left(3+ \frac{2}{\rho}+ {\hat \varepsilon}  \right) \sum\limits_{x \in  \mathcal{X} \setminus  P^*  }  {\rm cost}_c^* (x)
-  \sum\limits_{x \in  \mathcal{X} \setminus   P }   {\rm cost}_c (x)
\nn \\
& &
+\ 2 \left(1+ \frac{1}{\rho} \right) \sqrt{ \sum\limits_{x \in  \mathcal{X} \setminus   P^* } {\rm cost}_c^* (x)}
\sqrt{\sum\limits_{x \in  \mathcal{X} \setminus   P   } {\rm cost}_c (x)}
+ \left(1+ \frac{1}{\rho} \right) \sum\limits_{x \in P^*}  p_x
-\sum\limits_{x \in   P  }    p_x
\nn   \\
& = &
\left(3+ \frac{2}{\rho} + {\hat \varepsilon} \right) {\rm cost}_c^* - {\rm cost}_c + 2 \left(1+ \frac{1}{\rho} \right) \sqrt{{\rm cost}_c^* C_s}
+ \left(1+ \frac{1}{\rho} \right)  {\rm cost}_p^*  - {\rm cost}_p
\nn \\
& \le &
\left(\left(3+ \frac{2}{\rho}+ {\hat \varepsilon} \right)  {\rm cost}_c^* + \left(1+ \frac{1}{\rho} \right)  {\rm cost}_p^*  \right)
- \left( {\rm cost}_c +  {\rm cost}_p \right)
\nn \\
& &
+ \frac{2 \left(1+ \frac{1}{\rho} \right) }{\sqrt{ 3+ \frac{2}{\rho} + {\hat \varepsilon} } }
\sqrt{\left(\left(3+ \frac{2}{\rho} + {\hat \varepsilon} \right)  {\rm cost}_c^* + \left(1+ \frac{1}{\rho} \right)  {\rm cost}_p^*  \right)
\left( {\rm cost}_c +  {\rm cost}_p \right)}.
\label{ieq2-thm-penalty-kmeap}
\eea
The RHS of \reff{ieq2-thm-penalty-kmeap} is equal to
\bea
& & \left(
\sqrt{ \left(3+ \frac{2}{\rho} + {\hat \varepsilon} \right)  {\rm cost}_c^* + \left(1+ \frac{1}{\rho} \right)  {\rm cost}_p^*   }
+ \alpha \sqrt{{\rm cost}_c +  {\rm cost}_p}
\right)
\nn \\
& &
\times
\left(
\sqrt{ \left(3+ \frac{2}{\rho} + {\hat \varepsilon}\right)  {\rm cost}_c^* + \left(1+ \frac{1}{\rho} \right)  {\rm cost}_p^*   }
-
\beta
\sqrt{{\rm cost}_c +  {\rm cost}_p}
\right)
\nn
\eea
where
\bea
\alpha = \frac{1+1/\rho}{\sqrt{3+2/\rho+\hat{\veps}}} + \sqrt{\frac{(1+1/\rho)^2}{3+2/\rho+\hat{\veps}} + 1 - \veps}, \nn\\
\beta = -\frac{1+1/\rho}{\sqrt{3+2/\rho+\hat{\veps}}} + \sqrt{\frac{(1+1/\rho)^2}{3+2/\rho+\hat{\veps}} + 1 - \veps}.\nn
\eea
This implies that
\bea
\sqrt{ \left(3+ \frac{2}{\rho} + {\hat \varepsilon}\right)  {\rm cost}_c^* + \left(1+ \frac{1}{\rho} \right)  {\rm cost}_p^*   }
-
\beta
\sqrt{{\rm cost}_c +  {\rm cost}_p} \ge 0,
\nn
\eea
which is equivalent to
\bea
{\rm cost}_c +  {\rm cost}_p
& \le &
\frac{1}{\beta^2} \left( 3+ \frac{2}{\rho} + {\hat \varepsilon}\right) {\rm cost}_c^*
+ \frac{1}{\beta^2} \left(1+ \frac{1}{\rho} \right) {\rm cost}_p^*.
\nn
\label{ieq3-thm-penalty-kmeap}
\eea
Observe that
$$
\frac{1}{\beta^2} \le 3+\frac{2}{\rho} + \hat{\veps}.
$$
So, we have
\bea
{\rm cost}_c +  {\rm cost}_p
\le
 \left(   3+ \frac{2}{\rho} + {\hat \varepsilon}  \right)^2 {\rm cost}_c^*
+   \left(3+ \frac{2}{\rho} + {\hat \varepsilon} \right)
\left(1+ \frac{1}{\rho} \right)
 {\rm cost}_p^*.
\label{ieq4-thm-penalty-kmeap}
\eea
Substituting ${\hat \varepsilon} = 1 / \rho$ into \reff{ieq4-thm-penalty-kmeap} completes the proof.
\endproof

We remark that Algorithm \ref{alg-multi-swap} can be adapted to a polynomial-time algorithm that only sacrifices   $ \varepsilon $ in the approximation factor (see \citealt{agkmmp}). Combining this adaptation and Theorems \ref{thm-multi-swap-kmedp}-\ref{thm-multi-swap-kmeap}, we obtain a $(3+\veps)$-approximation algorithm for $k$-MedP, and a $(9+\veps)$-approximation algorithm for $k$-MeaP, if $\rho$ and $\hat{\veps}$ are sufficiently small.

\section{Local search algorithm for $k$-MedO/$k$-MeaO}\label{sec:outlier}
In this section, we focus on $k$-MedO and $k$-MeaO. We apply the technique for addressing outliers in a local search algorithm provided by \cite{gklmv} to $k$-MedO and $k$-MeaO, and use our new analysis to improve the approximation ratio.

\subsection{The algorithm}
Each iteration of the outlier-based multi-swap local search algorithm has a no-swap step and a swap step. Supposing that the current solution is $(S,P)$, the no-swap step implements an ``add outliers'' operation that adds the points in $\outlier(S,P)$ (defined in Section \ref{sec:model}) to $P$, if this operation can reduce the cost by a given factor. Then, the swap step searches for a better solution by the multi-swap together with the ``add outliers'' operations. The algorithm terminates when both the no-swap  and the swap step can not reduce the cost by the given factor.

Let ${\rm cost}(S,P)$ denote the the cost of the solution $(S,P)$. We give the formal description of the outlier-based local search algorithm in Algorithm \ref{alg-multi-swap-outlier}.
This algorithm has three parameters: $\rho$ is the number of points which are allowed to be swapped in a solution, $q$ and $\veps$ are used to control the descending step-length of the cost. The parameter $q$ is fixed as $k$ in the algorithm provided by \cite{gklmv}, while it is an input in our algorithm, because the approximation ratio is associated with the value of this parameter.

The following proposition holds for this algorithm.
\begin{proposition}[\citealt{gklmv}]\label{prop1}
Let $(S,P)$ be the solution produced by LS-Multi-Swap-Outlier($\cX,$ $C,z,k,\rho,\veps$), and set $q = k$ if $\rho = 1$, otherwise, set $q = k^2-k$. Then
\begin{itemize}
  \item[{\rm(i)}] $\cost(S, P \cup \outlier(S,P)) \ge  \left( 1-\veps/q \right) \cost(S,P)$,
  \item[{\rm(ii)}] $\cost(S \setminus A \cup B, P \cup \outlier(S \setminus A \cup B,P)) \ge \left( 1-\veps/q \right) \cost(S,P) $ for any $A \subseteq S$ and $B \subseteq C$.
\end{itemize}
\end{proposition}
\begin{breakablealgorithm}
\renewcommand{\algorithmicrequire}{\textbf{Input:}}
\renewcommand{\algorithmicensure}{\textbf{Output:}}
\caption{The outlier-based local search algorithm: LS-Multi-Swap-Outlier($\cX,C,z,k,\rho,q,\veps$)}
\label{alg-multi-swap-outlier}
\begin{algorithmic}[1]
    \REQUIRE Data set $\cX$, candidate center set $C$, positive integers $z$, $k$, $q$ and $\rho \le k$, real number $\veps > 0$.
    \ENSURE Center set $S \subseteq C$ and outlier set $P \subseteq \cX$.
    \STATE Arbitrarily choose a $k$-center subset $S$ from $C$.
    \STATE Set $P := \outlier(C)$.
    \STATE Set $\alpha := +\infty$.
    \WHILE {$\cost(S,P) < \alpha$}
        \STATE $\alpha \leftarrow \cost(S,P)$
        \IF {$\cost(S, P \cup \outlier(S,P)) < \left( 1-\dfrac{\veps}{q} \right)\cost(S,P)$}
            \STATE Set $P := P \cup \outlier(S,P)$.
        \ENDIF
        \STATE Compute
            $(A,B): = \arg\min_{A\subseteq S, B \subseteq C\setminus S, |A|=|B| \le \rho}
            {\rm cost} (S\setminus A \cup B, P\cup \outlier(S \setminus A \cup B, P)).$
        \STATE Set $S' := S \setminus A \cup B$ and $P' := P \cup \outlier(S \setminus A \cup B,P)$.
        \IF {$\cost(S',P') < \left( 1-\dfrac{\veps}{q} \right)\cost(S,P)$}
            \STATE Set $S := S'$ and $P := P'$.
        \ENDIF
    \ENDWHILE
    \RETURN $S$ and $P$
\end{algorithmic}
\end{breakablealgorithm}

For $k$-MedO, we run LS-Multi-Swap-Outlier($\cX,\cF,z,k,\rho,q,\veps$). For $k$-MeaO, we run LS-Multi-Swap-Outlier($\cX,\cC',z,k,\rho,q,\veps$), where $\cC'$ is an ${\hat \varepsilon}$-approximate centroid set for $\cX$. The values of $\rho$, $\veps$, and $\hat{\veps}$ will be determined in the analysis of the algorithm.

\subsection{The analysis}
The  time complexity of Algorithm \ref{alg-multi-swap-outlier} is shown in the following theorem.
\begin{theorem} \label{thm-time-outlier}
The running time of {\rm LS-Multi-Swap-Outlier($\cX,\cC,z,k,\rho,q,\veps$)} is $O\left(\frac{k^{\rho}n^{\rho}q}{\veps} \log(n\delta)\right)$.
\end{theorem}
\proof
The proof  is similar to that in \cite{gklmv}. For the sake of completeness, we present it also here.
W.l.o.g, we can assume that the optimal value of the problem is
at least $1$ by scaling the distances, except for the trivial case that $k = n - z$.
Under this assumption, the cost of any solution is  at most $n\delta \ge 1$. The number of iterations is at most $O(- \log_{1-\veps/q}(n\delta)) = O(\frac{q}{\veps}\log(n\delta))$, since the cost is reduced to at most $(1- \veps/ q)$ times the old cost in each iteration. The number of solutions searched by a swap operation  is at most $O((kn)^\rho)$, since $|A|=|B|\le \rho$. This completes the proof.
\endproof

The algorithm may violates the outlier constraint in order to yield a bounded approximation ratio. We can also bound the number of outliers by a suitable factor, which is shown in the following result.
\begin{theorem}\label{thm-number-outlier}
The number of outliers returned by {\rm LS-Multi-Swap-Outlier($\cX,\cC,$ $z, k,\rho,q,\veps$)} is $O\left(\frac{zq}{\veps}\right.$ $\left.\log(n\delta)\right)$.
\end{theorem}
\proof
From the proof of Theorem \ref{thm-time-outlier}, we know that LS-Multi-Swap-Outlier($\cX,\cC,$ $z,k,\rho,q,\veps$) has at most $O\left(\frac{q}{\veps}\log(n\delta)\right)$ iterations.  In each iteration, the algorithm removes at most $2z$ outliers. This completes the proof.
\endproof

Let $(S,P)$ be the solution returned by Algorithm \ref{alg-multi-swap-outlier}, and $(S^*,P^*)$ be the global optimal solution.
Similar to the penalty version, we use the same notations (except that the outlier version has not penalty cost) and adopt the same partition of $S$ and $S^*$ ($S= \cup_l S_l,\ S^*= \cup_l S^*_l$).
Similar to Lemmas \ref{lem-penalty-case1} and \ref{lem-penalty-case2}, we obtain the following two results.
\begin{lemma}\label{lem-outlier-case1}
If $ |S_l| = |S^*_l| \le \rho$, we have
\bea
- \dfrac{\veps}{q} \cdot \cost(S,P) & \le &
\sum\limits_{s \in  S_l }  \sum\limits_{x \in  N(s) \setminus  P^* }
\left(d^2 (\phi( \cent_{\cC}(N^*_q (s^*_x) )), x)  - {\rm cost}_c (x) \right) +
\nn  \\
& &
\sum\limits_{s^* \in  S_l^* } \sum\limits_{x \in   N^* (s^*)} {\rm cost}_c^* (x)
- \sum\limits_{s^* \in  S_l^* } \sum\limits_{x \in   N^* (s^*) \setminus  P}{\rm cost}_c (x). \label{ieq1-lem-outlier-case1}
\eea
for $k$-MedO, and
\bea
- \dfrac{\veps}{q} \cdot \cost(S,P) & \le &
\sum\limits_{s \in  S_l }  \sum\limits_{x \in  N(s) \setminus  P^* }
\left(  d^2 (\phi( \cent_{\cC}(N^*_q (s^*_x) )), x)  - {\rm cost}_c (x) \right) +
\nn  \\
& &
\sum\limits_{s^* \in  S_l^* } \sum\limits_{x \in   N^* (s^*)} (1+ {\hat \varepsilon})  {\rm cost}_c^* (x)
- \sum\limits_{s^* \in  S_l^* } \sum\limits_{x \in   N^* (s^*) \setminus  P}{\rm cost}_c (x). \label{ieq2-lem-outlier-case1} \nn\\
\eea
for $k$-MeaO.
\end{lemma}
\proof
We only prove it for $k$-MeaO. The proof for $k$-MedO is similar.
We consider the swap$(S_l,{\hat S}^*_l)$. Since the swap step of the algorithm produces at most $z$ additional outliers in each iteration, and $|P \setminus \bigcup_{s^* \in S^*_l} N^* (s^*) \cup P^*| \le |P|+z$, we can let the points in $P \setminus \bigcup_{s^* \in S^*_l} N^* (s^*) \cup P^*$ be the additional outliers after the constructed swap operation. For the other points, it is obvious that we can apply the reassignments  in the proof of Lemma \ref{lem-penalty-case1} also here.
Then, Proposition \ref{prop1} yields
\bea
&& - \dfrac{\veps}{q} \cdot \cost(S,P) \nn\\
& \le &
{\rm cost}(S \setminus S_l \cup \hat{S^*_l}, P \cup \outlier(S \setminus S_l \cup \hat{S^*_l},P) ) - {\rm cost}(S,P)
\nn \\
& \le &
- \sum\limits_{s \in S_l } \sum\limits_{x \in  N(s)\cap P^*} {\rm cost}_c (x) \nn\\
&&
+ \sum\limits_{s \in  S_l } \sum\limits_{x \in  N(s)  \setminus \left( \bigcup_{s^* \in S^*_l} N^* (s^*) \cup P^*\right)   }
\left( d^2 (\phi( \cent_{\cC}(N^*_q (s^*_x) )), x)   - {\rm cost}_c (x) \right)
\nn \\
& &
+ \sum\limits_{s^* \in  S^*_l }  \sum\limits_{x \in   N^* (s^*) \setminus P} ( d^2 ({\hat s^*}, x)  - {\rm cost}_c (x))
+ \sum\limits_{s^* \in  S^*_l }  \sum\limits_{x \in   N^* (s^*)  \cap P}  d^2 ({\hat s^*}, x)
\nn \\
& \le &
- \sum\limits_{s \in  S_l }  \sum\limits_{x \in  N(s)\cap P^*} {\rm cost}_c (x) + \sum\limits_{s^* \in  S_l^* } \sum\limits_{x \in   N^* (s^*)  \cap P} (1+ {\hat \varepsilon}) {\rm cost}_c^* (x) \nn\\
&&
+ \sum\limits_{s \in  S_l }  \sum\limits_{x \in  N(s) \setminus  P^* }
\left(  d^2 (\phi( \cent_{\cC}(N^*_q (s^*_x) )), x)  - {\rm cost}_c (x) \right)
\nn  \\
& &
+ \sum\limits_{s^* \in  S_l^* } \sum\limits_{x \in   N^* (s^*) \setminus P} ((1+ {\hat \varepsilon})  {\rm cost}_c^* (x)  - {\rm cost}_c (x))
\nn\\
& \le &
\sum\limits_{s \in  S_l }  \sum\limits_{x \in  N(s) \setminus  P^* }
\left(  d^2 (\phi( \cent_{\cC}(N^*_q (s^*_x) )), x)  - {\rm cost}_c (x) \right)
\nn  \\
& &
+ \sum\limits_{s^* \in  S_l^* } \sum\limits_{x \in   N^* (s^*)} (1+ {\hat \varepsilon})  {\rm cost}_c^* (x)
- \sum\limits_{s^* \in  S_l^* } \sum\limits_{x \in   N^* (s^*) \setminus P}{\rm cost}_c (x), \nn
\eea
where the third inequality follows from  \reff{ieq-approx-cent}.

\endproof

\begin{lemma}\label{lem-outlier-case2}
For any point $s \in S_l \setminus \{s_l\}$ and $s^* \in S^*_l$,  we have
\bea
- \dfrac{\veps}{q} \cdot \cost(S,P)
& \le &
\sum\limits_{x \in  N(s) \setminus  P^* }
\left(  d(\phi( \cent_{\cC}(N^*_q (s^*_x) )), x)  - {\rm cost}_c (x) \right)
\nn \\
& &
+ \sum\limits_{x \in   N^* (s^*)} {\rm cost}_c^* (x)
- \sum\limits_{x \in   N^* (s^*) \setminus P} {\rm cost}_c(x)
\label{ieq1-lem-outlier-case2}
\eea
for $k$-MedO, and
\bea
- \dfrac{\veps}{q} \cdot \cost(S,P)
& \le &
\sum\limits_{x \in  N(s) \setminus  P^* }
\left(  d^2 (\phi( \cent_{\cC}(N^*_q (s^*_x) )), x)  - {\rm cost}_c (x) \right)
\nn \\
& &
+ \sum\limits_{x \in   N^* (s^*)} (1+ {\hat \varepsilon}) {\rm cost}_c^* (x)
- \sum\limits_{x \in   N^* (s^*) \setminus P } {\rm cost}_c (x)
\label{ieq2-lem-outlier-case2}
\eea
for $k$-MeaO.
\end{lemma}
\begin{proof}
The proof is similar to those for Lemmas \ref{lem-penalty-case2} and \ref{lem-outlier-case1}.
\end{proof}

Next we will construct some swap operations for each pair $(S_l, S^*_l)$, and then apply Lemmas \ref{lem-outlier-case1} and \ref{lem-outlier-case2} to these swaps. Similar to the analysis for the penalty version, we consider two cases according to the size of $S_l$: $|S_l| \le \rho$ and $|S_l| =m_l > \rho$.

Note that the number of constructed swap operations will appear in the coefficient of $\cost(S,P)$ after summing the inequalities in Lemmas \ref{lem-outlier-case1} and \ref{lem-outlier-case2}. We want this number to be as small as possible, since it is proportional to the approximation ratio due to the later analysis. On the other hand, to obtain the entire cost of the solution $(S,P)$, we need to swap all centers in $S$ at least once. Thus, for the case of $|S_l| \le \rho$, we consider the same swap operations in the analysis for the penalty version (we state it again in the following Case 1), since each center in $S_l$ is swapped exactly once.

For the case of $|S_l| = m_l > \rho$, there are $m_l(m_l-1)$ single-swap operations in the analysis for the penalty version. This makes the coefficient of the cost of $(S^*,P*)$ small  ($m_l / (m_l-1) \rightarrow 1$ when $m_l \rightarrow +\infty$), but the number of swaps is large. In this section, we consider two methods to construct swap operations for this case, which are stated in Methods 1 and 2 in the following Case 2.
Note that Method 2 is the same as that in Section \ref{sec:penalty}.
\begin{description}
\item[Case 1] (cf. Figure \ref{fig:multiswap} for $\rho=3$). For each $l$ with $ |S_l| = |S^*_l| \le \rho$, let $S_l=\{s_l\}$ and $S^*_l=\{s^*_{l}\}$.
We construct the swap$(s_l, s^*)$ for $k$-MedO, and swap$(s_l, {\hat s^*_l}) $ for $k$-MeaO.
\item[Case 2.] For each $l$ with $ |S_l| = |S^*_l| = m_l > 1$, let $S_l=\{s_l,s_{l,2},\dots,s_{l,m_l}\}$ and $S^*_l=\{s^*_{l,1},s^*_{l,2},\dots,s^*_{l,m_l}\}$.
    \begin{itemize}
      \item[Method 1] (cf. Figure \ref{fig:singleswap2}). Set
        $$
        \psi (s^*) :=
        \left\{\ba{ll}
        s_{l,2},  &  {\rm if} \  s^* = s^*_{l,1}; \\
        s_{l,2},  &  {\rm if} \  s^* = s^*_{l,2}; \\
        s_{l,3},  &  {\rm if} \  s^* = s^*_{l,3}; \\
        ... & ... \\
        s_{l,m_l},  &  {\rm if} \  s^* = s^*_{l,m_l}.
        \ea\right.
        $$
        For each $s^* \in S^*_l$, we construct swap$(\psi(s^*),s^*)$ for $k$-MedO, and swap$(\psi($ $s^*),{\hat s^*})$ for $k$-MeaO.
      \item[Method 2] (cf. Figure \ref{fig:singleswap1}). We consider $ (m_l-1)m_l$ pairs $(s, s^*)$ with
        $s \in S_l \backslash \{ s_l\} $ and $ s^* \in S^*_l$. For $k$-MedO, we construct swap$(s, s^*)$ for each pair; for $k$-MeaO, we construct swap$(s, {\hat s^*})$ for each pair.
    \end{itemize}
\end{description}

\begin{figure} \label{fig:singleswap2}
  \centering
  \includegraphics[width=3.0cm]{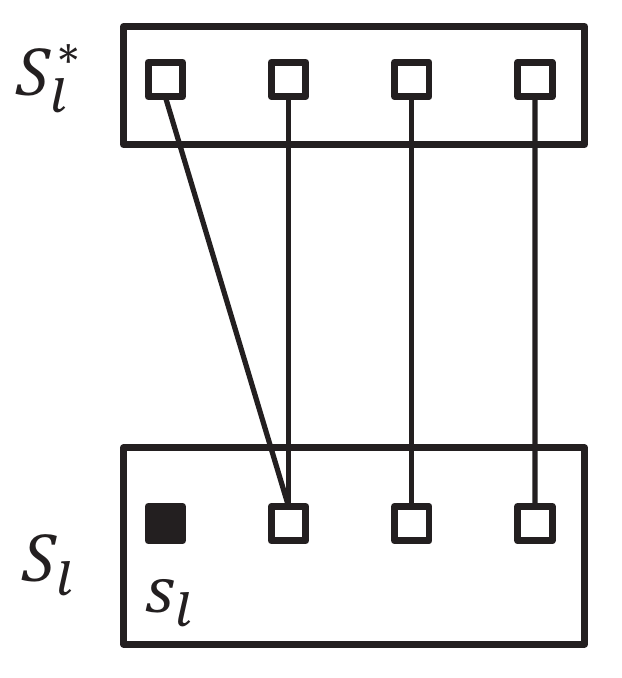}
  \caption{Sinlgle-swap operations for the case of $|S_l| > \rho$.}
\end{figure}

Combining these swap operations, we obtain the main results for Algorithm \ref{alg-multi-swap-outlier}, which are shown in the following two theorems.

\begin{theorem} \label{thm-multi-swap-kmedo}
Let $(S,P)$ be the solution returned by {\rm LS-Multi-Swap-Outlier($\cX,$ $\cF,z,k,\rho,q,\veps$)} for $k$-MedO. If $(1+k)\veps < q$, then we have
\bea\label{result1-thm-multi-swap-kmedo}
\cost(S,P) \le \frac{5}{1- (1+k)\veps / q} \cdot \cost(S^*,P^*).
\eea
If $(1+k^2-k)\veps < q$, then we have
\bea\label{result2-thm-multi-swap-kmedo}
\cost(S,P) \le \frac{3+2/\rho}{1- (1+k^2-k)\veps / q} \cdot \cost(S^*,P^*).
\eea
\end{theorem}

\begin{theorem} \label{thm-multi-swap-kmeao}
Let $\cC'$ be an $\hat{\veps}$-approximate centroid set for the data set $\cX$, and $(S,P)$ be the solution returned by {\rm LS-Multi-Swap-Outlier($\cX,\cC',z,k,\rho,q,\veps$)} for $k$-MeaO. If $(5+\hat{\veps})(1+k)\veps < (9+\hat{\veps})q$, then we have
\bea\label{result1-thm-multi-swap-kmeao}
\cost(S,P) \le \dfrac{5+\hat{\veps}}{\beta_1^2} \cdot \cost(S^*,P^*)
\eea
where
$$
\beta_1 = -\frac{2}{\sqrt{5 + \hat{\veps}}} + \sqrt{\frac{4}{5 + \hat{\veps}}+1- \frac{(1+k)\veps}{q}}.
$$
If $(1+k^2-k)\veps / q < (1+1/\rho)^2 / (3+ 2/\rho + \hat{\veps})+1$, then we have
\bea\label{result2-thm-multi-swap-kmeao}
\cost(S,P) \le \dfrac{3+2/\rho+\hat{\veps}}{\beta_2^2} \cdot \cost(S^*,P^*)
\eea
where
$$
\beta_2 = -\frac{1+1/\rho}{\sqrt{3+ 2/\rho + \hat{\veps}}} + \sqrt{\frac{(1+1/\rho)^2}{3+ 2/\rho + \hat{\veps}}+1- \frac{(1+k^2-k)\veps}{q}}.
$$
\end{theorem}

Each of these two theorems gives two approximation ratios for Algorithm \ref{alg-multi-swap-outlier}. The first one is obtained by Method 1, while the second one is obtained by Method 2.

\proof[Proof of Theorem \ref{thm-multi-swap-kmedo}.]
We first prove  inequality \reff{result1-thm-multi-swap-kmedo}. For Case 2, we use Method 1 to construct swap operations.
Note that each point in $S$ is swapped at most twice, and each point in $S^*$ is swapped once exactly, implying that the number of constructed swap operations is $|S^*|=k$. Summing inequality \reff{ieq1-lem-outlier-case2} over these $k$ swaps and using Proposition \ref{prop1}, we obtain
\bea\label{ieq1-thm-single-kmedo}
- \frac{k\veps}{q} \cdot \cost(S,P) & \le &
2 \sum\limits_{s \in S} \sum\limits_{x \in  N(s) \setminus  P^* }
\left(  d(\phi( \cent_{\cC}(N^*_q (s^*_x) )), x)  - {\rm cost}_c (x) \right)
\nn \\
& &
+ \sum\limits_{s^* \in S^*} \left( \sum\limits_{x \in N^* (s^*)} {\rm cost}_c^* (x)
- \sum\limits_{x \in N^* (s^*) \setminus P} {\rm cost}_c (x) \right)
\nn \\
& \le &
2 \sum\limits_{x \in \cX \setminus (P \cup P^*)}
\left(  d(\phi( \cent_{\cC}(N^*_q (s^*_x) )), x)  - {\rm cost}_c (x) \right)
\nn \\
& &
+ \sum\limits_{x \in \cX \setminus P^*} \cost_c (x) - \sum\limits_{x \in \cX \setminus P} \cost_c(x) + \sum\limits_{P^* \setminus P} \cost_c(x)
\nn\\
& \le &
4\cost(S^*,P^*) + \cost(S^*,P^*) - \cost(S,P) + \sum\limits_{P^* \setminus P} \cost_c(x)
\nn\\
& = &
5\cost(S^*,P^*) - \cost(S,P) + \sum\limits_{P^* \setminus P} \cost_c(x),
\eea
where the third inequality follows from \reff{ieq2-thm-penalty-kmedp} and \reff{ieq4-thm-penalty-kmedp}.

Using the definition of $\outlier(\cdot,\cdot)$, we obtain
\bea\label{ieq-ub-by-def-out}
\sum\limits_{x \in P^* \setminus P} \cost_c(x)
& \le & \sum\limits_{x \in \outlier(S,P)} \cost_c(x) \nn\\
& = & \cost(S,P) - \cost(S,P\cup \outlier(S,P)) \nn\\
& \le &  \dfrac{\veps}{q} \cdot \cost(S,P).
\eea

Combining inequalities  \reff{ieq1-thm-single-kmedo}-\reff{ieq-ub-by-def-out}, we have
\bea
0 \le 5{\rm cost}(S^*,P^*) - \left(1- \dfrac{(1+k)\veps}{q} \right){\rm cost}(S,P), \nn
\eea
which is equivalent to \reff{result1-thm-multi-swap-kmedo} under the condition that $(1+k)\veps < q$.

Next we will prove the inequality \reff{result2-thm-multi-swap-kmedo}. For Case 2, we use Method 2 to construct swap operations.
Let $L_1 := \{ l~|~ |S_l| \le \rho \}$ and $L_2 := \{ l~|~|S_l| > \rho \}$.
Summing  inequality \reff{ieq1-lem-outlier-case1} with weight 1 and inequality \reff{ieq1-lem-outlier-case2} with weight $1/ (m_l-1)$ over all constructed swap operations, and observing that $m_l/(m_l-1) \le (\rho+ 1)/\rho $, we obtain
\bea\label{ieq1-thm-outlier-kmedo}
& & - \sum\limits_{l \in L_1} \dfrac{\veps}{q} \cdot \cost(S,P) - \sum\limits_{l \in L_2} \frac{1}{m_l-1} \cdot m_l(m_l-1) \cdot \dfrac{\veps}{q} \cdot \cost(S,P)
\nn\\
& \le &
\left(1+ \frac{1}{\rho} \right)
\sum\limits_{s \in S}  \sum\limits_{x \in  N(s) \setminus  P^* }
\left(  d^2 (\phi( \cent_{\cC}(N^*_q (s^*_x) )), x)   - {\rm cost}_c (x) \right)
\nn \\
& &
+   \sum\limits_{s^* \in S^*}  \sum\limits_{x \in   N^* (s^*)}
{\rm cost}_c^* (x)
-  \sum\limits_{s^* \in S^*}  \sum\limits_{x \in   N^* (s^*) \setminus  P}
{\rm cost}_c (x).
\eea
Note that there are at most $k(k-1)$ constructed swap operations. It follows from $1 / (m_l-1) \le 1$ that
\bea\label{ieq2-thm-outlier-kmedo}
{\rm LHS~of~\reff{ieq1-thm-outlier-kmedo}}
& \ge & - \left( |L_1| + \sum\limits_{l \in L_2} m_l(m_l-1) \right) \cdot \dfrac{\veps}{q} \cdot \cost(S,P)
\nn\\
& \ge & \dfrac{(k^2-k)\veps}{q} \cdot \cost(S,P).
\eea
Inequality \reff{ieq2-thm-penalty-kmedp} then yields the following upper bound for the RHS of \reff{ieq1-thm-outlier-kmedo}.
\bea\label{ieq3-thm-outlier-kmedo}
{\rm RHS~of~\reff{ieq1-thm-outlier-kmedo}}
& \le &
\left(3+ \frac{2}{\rho}\right) \sum\limits_{x \in  \mathcal{X} \setminus  P^*  }  {\rm cost}_c^* (x)
-  \sum\limits_{x \in  \mathcal{X} \setminus  P} {\rm cost}_c (x) + \sum\limits_{x \in P^* \setminus P} \cost_c(x)
\nn \\
& = &
\left(3+ \frac{2}{\rho}\right) {\rm cost}_c^* - {\rm cost}_c +  \sum\limits_{x \in P^* \setminus P} \cost_c(x).
\eea

Combining inequalities \reff{ieq-ub-by-def-out}-\reff{ieq3-thm-outlier-kmedo}, we have
\bea
0 \le \left(3+ \frac{2}{\rho}\right) {\rm cost}(S^*,P^*) - \left(1-\dfrac{(1+k^2-k)\veps}{q}  \right){\rm cost}(S,P), \nn
\eea
which is equivalent to \reff{result2-thm-multi-swap-kmedo} under the condition $(1+k^2-k)\veps < q$.

\endproof

\proof[Proof of Theorem \ref{thm-multi-swap-kmeao}.]
We first use Method 1 for Case 2 to prove \reff{result1-thm-multi-swap-kmeao}. Similar to the proof for $k$-MedO, we have
\bea\label{ieq1-thm-single-kmeao}
- \dfrac{k\veps}{q} \cdot \cost(S,P) & \le &
2 \sum\limits_{x \in \cX \setminus (P \cup P^*)}
\left(  d^2(\phi( \cent_{\cC}(N^*_q (s^*_x) )), x)  - {\rm cost}_c (x) \right)
\nn \\
& &
+ \sum\limits_{x \in \cX \setminus P^*} (1+\hat{\veps}) \cost_c (x) - \sum\limits_{x \in \cX \setminus P} \cost_c(x) + \sum\limits_{P^* \setminus P} \cost_c(x)
\nn\\
& \le &
4 \sum\limits_{x \in  \mathcal{X}  \setminus (P \cup P^*) } {\rm cost}^*_c (x)
\nn\\
&&
+ 4 \sqrt{ \sum\limits_{x \in  \mathcal{X}  \setminus (P \cup P^*) }  {\rm cost}^*_c (x) }
\cdot  \sqrt{ \sum\limits_{x \in  \mathcal{X}  \setminus (P \cup P^*) }  {\rm cost}_c (x) }
\nn\\
&&
+ \sum\limits_{x \in \cX \setminus P^*} (1+\hat{\veps}) \cost_c (x) - \sum\limits_{x \in \cX \setminus P} \cost_c(x) + \sum\limits_{P^* \setminus P} \cost_c(x)
\nn\\
& \le &
4 \sqrt{ \cost(S^*,P^*) } \cdot  \sqrt{ \cost(S,P) }
\nn\\
&&
+ (5+\hat{\veps}) \cost(S^*,P^*) - \cost(S,P) +  \dfrac{\veps}{q} \cdot \cost(S,P),
\eea
where the second inequality follows from Lemma \ref{lem-ub-cluster form} (this lemma still holds for the outlier version of $k$-means), and the third inequality follows from \reff{ieq-ub-by-def-out}.

When $(5+\hat{\veps})(1+k)\veps < (9+\hat{\veps})q$, it follows by factorization that inequality \reff{ieq1-thm-single-kmeao} is equivalent to
\bea\label{ieq-factorization-single}
0 & \le & \left( \sqrt{\left(5 + \hat{\veps}\right)\cost(S^*,P^*)} + \alpha \sqrt{\cost(S,P)} \right) \nn\\
&& \times \left( \sqrt{\left(5 + \hat{\veps}\right)\cost(S^*,P^*)} - \beta_1 \sqrt{\cost(S,P)} \right),
\eea
where
\bea
\alpha &=& \frac{2}{\sqrt{5 + \hat{\veps}}} + \sqrt{\frac{4}{5 + \hat{\veps}}+1-\frac{(1+k)\veps}{q}}, \nn\\
\beta_1 &=& -\frac{2}{\sqrt{5 + \hat{\veps}}} + \sqrt{\frac{4}{5 + \hat{\veps}}+1-\frac{(1+k)\veps}{q}}. \nn
\eea
Since the first term of the RHS of \reff{ieq-factorization-single} is non-negative, we obtain
$$
\sqrt{\left(5 + \hat{\veps}\right)\cost(S^*,P^*)} - \beta_1 \sqrt{\cost(S,P)} \ge 0,
$$
which gives \reff{result1-thm-multi-swap-kmeao}.

Next we prove inequality \reff{result2-thm-multi-swap-kmeao}. For Case 2, we use Method 2 to construct swap operations.
Similar to the proof of Theorem \ref{thm-multi-swap-kmedo}, summing inequality  \reff{ieq2-lem-outlier-case1} with weight 1 and inequality \reff{ieq2-lem-outlier-case2} with weight $1/ (m_l-1)$ over all constructed swap operations implies that
\bea\label{ieq1-thm-outlier-kmeao}
& & - \sum\limits_{l \in L_1} \dfrac{\veps}{q} \cdot \cost(S,P) - \sum\limits_{l \in L_2} \frac{1}{m_l-1} \cdot m_l(m_l-1) \cdot \dfrac{\veps}{q} \cdot \cost(S,P)
\nn\\
& \le &
\left(1+ \frac{1}{\rho} \right)
\sum\limits_{s \in S}  \sum\limits_{x \in  N(s) \setminus  P^* }
\left(  d^2 (\phi( \cent_{\cC}(N^*_q (s^*_x) )), x)   - {\rm cost}_c (x) \right)
\nn \\
& &
+   \sum\limits_{s^* \in S^*}  \sum\limits_{x \in   N^* (s^*)}
(1+ {\hat \varepsilon}) {\rm cost}_c^* (x)
-  \sum\limits_{s^* \in S^*}  \sum\limits_{x \in   N^* (s^*) \setminus  P}
{\rm cost}_c (x).
\eea
Because of Lemma \ref{lem-ub-cluster form}, the RHS of \reff{ieq1-thm-outlier-kmeao} is bounded from above by
\bea\label{ieq2-thm-outlier-kmeao}
{\rm RHS~of~\reff{ieq1-thm-outlier-kmeao}}
& \le &
\left(3+ \frac{2}{\rho}+ {\hat \varepsilon}  \right) \sum\limits_{x \in  \mathcal{X} \setminus  P^*  }  {\rm cost}_c^* (x)
-  \sum\limits_{x \in  \mathcal{X} \setminus  P} {\rm cost}_c (x) + \sum\limits_{x \in P^* \setminus P} \cost_c(x)
\nn \\
& &
+\ 2 \left(1+ \frac{1}{\rho} \right) \sqrt{ \sum\limits_{x \in  \mathcal{X} \setminus  P^* } {\rm cost}_c^* (x)}
\sqrt{\sum\limits_{x \in  \mathcal{X} \setminus  P} {\rm cost}_c (x)}
\nn \\
& = &
\left(3+ \frac{2}{\rho} + {\hat \varepsilon} \right) \cost(S^*,P^*) - \cost(S,P)  + \sum\limits_{x \in P^* \setminus P} \cost_c(x)
\nn\\
&&
+\ 2 \left(1+ \frac{1}{\rho} \right) \sqrt{\cost(S^*,P^*) \cost(S,P)}.
\eea

Combining inequalities \reff{ieq-ub-by-def-out}, \reff{ieq2-thm-outlier-kmedo},  \reff{ieq1-thm-outlier-kmeao} and \reff{ieq2-thm-outlier-kmeao}, we have
\bea
0 & \le & \left(3+ \frac{2}{\rho} + {\hat \varepsilon} \right) {\rm cost}(S^*,P^*) - \left(1-\frac{(1+k^2-k)\veps}{q} \right){\rm cost}(S,P)
\nn\\
&&
+\ 2 \left(1+ \frac{1}{\rho} \right) \sqrt{{\rm cost}(S^*,P^*) \cost(S,P)}. \nn
\eea
Using the factorization for this inequality and the condition in this theorem, we obtain the desired result.

\endproof

Consequently, we have the following corollaries that specify the tradeoff between the approximation ratio and the outlier blowup.
\begin{corollary}
There exists a bi-criteria $(5 + \veps, O(\frac{k}{\veps}\log(n\delta)))$-, and a bi-criteria $(3 + \veps, O(\frac{k^2}{\veps}\log(n\delta)))$-approximation algorithm for $k$-MedO.
\end{corollary}
\proof
If $q \ge k+1$, then
$$
\frac{5}{1- (1+k)\veps / q} \le \frac{5}{1- \veps} \sim 5 + O(\veps).
$$
If $q \ge k^2-k+1$ and $\rho \ge 2 / O(\veps)$, then
$$
\frac{3 + 2/\rho}{1- (1+k^k-k)\veps / q} \le \frac{3 + O(\veps)}{1- \veps} \sim 3 + O(\veps).
$$
Combining the above results, Theorems \ref{thm-number-outlier} and \ref{thm-multi-swap-kmedo} complete the proof.
\endproof

\begin{corollary}
There exists a bi-criteria $(25 + \veps, O(\frac{k}{\veps}\log(n\delta)))$-, and a bi-criteria $(9 + \veps, O(\frac{k^2}{\veps}\log(n\delta)))$-approximation algorithm for $k$-MeaO.
\end{corollary}
\proof
Recall the definitions of $\beta_1$ and $\beta_2$ in Theorem \ref{thm-multi-swap-kmeao}. We then have
$$
\dfrac{5+\hat{\veps}}{\beta_1^2} \sim 25 + O(\veps+\hat{\veps})
$$
when $q = k+1$,
and
$$
\dfrac{3+2/\rho+\hat{\veps}}{\beta_2^2} \sim 9 + O(\veps+\hat{\veps})
$$
when $q = k^2-k+1$ and $\rho$ is sufficiently large.
Then, combining these results, Theorems \ref{thm-number-outlier} and \ref{thm-multi-swap-kmeao} completes the proof.
\endproof

\section{Conclusions}
The previous analyses  of local search algorithms for the robust $k$-median/$k$-means, use only the individual form, in which the constructed connections between the local and global optimal solutions are individual for each point. This has the disadvantage that the joint information about outliers remains hidden. In this paper, we develop a cluster form analysis and define the adapted cluster that captures the outlier information.
We find that this new technique works better than the previous analysis methods of local search algorithms, since it improves the approximation ratios of local search algorithms for $k$-MeaP, $k$-MeaO and $k$-MedO, and obtain the same ratio which is the best for $k$-MedP.

We believe that our new technique will also work for the robust FLP, since the structure of FLP is similar to $k$-median/$k$-means. Also, our technique seems to be promising for the robust $k$-center problem, even for any algorithm for robust clustering problems that is based on local search.

\section{Acknowledgments}
The first author is supported by the NSFC under Grant No. 12001039. The second author is supported by the Science Foundation of the Anhui Education Department under Grant No. KJ2019A0834. The third author is supported by the NSFC under Grant No. 11971349. The fourth and fifth authors are supported by  the NSFC under Grant No. 11871081.  The fourth author is also supported by Beijing Natural Science Foundation Project under Grant No. Z200002.


\end{document}